\documentclass{IEEEtran}

\usepackage{eulervm}

\usepackage{amsmath}
\usepackage{amssymb}
\usepackage{amsthm}
\usepackage{graphicx}
\usepackage{subfigure}
\usepackage{lettrine}
\usepackage{array}
\usepackage{caption}
\usepackage[below]{placeins}

\date{}
\title{A One-Hop Information Based Geographic Routing Protocol for Delay Tolerant MANETs}

\author{Lei~You$^{1}$,
        Jianbo~Li$^{2*}$, 
        Changjiang Wei$^{3}$ 
        Chenqu Dai$^{4}$\\
\small{Information Engineering College, Qingdao University, Qingdao 266071, Shandong Province, China} \\
\small\texttt{youleiqdu@gmail.com$^1$,lijianboqdu@gmail.com$^2$,wcj@qdu.edu.cn$^3$, daichenqu@gmail.com$^4$} \\
}

\begin{document}

\maketitle

\theoremstyle{plain}
\newtheorem{definition}{Definition}
\newtheorem{theorem}{Theorem}

\def\received{}

\begin{abstract}
Delay and Disruption Tolerant Networks (DTNs) may lack continuous network connectivity. Routing in DTNs is thus a challenge since it must handle network partitioning, long delays, and dynamic topology. Meanwhile, routing protocols of the traditional Mobile Ad hoc NETworks (MANETs) cannot work well due to the failure of its assumption that most network connections are available. In this article, a geographic routing protocol is proposed for MANETs in delay tolerant situations, by using no more than one-hop information. A utility function is designed for implementing the under-controlled replication strategy. To reduce the overheads caused by message flooding, we employ a criterion so as to evaluate the degree of message redundancy. Consequently a message redundancy coping mechanism is added to our routing protocol. Extensive simulations have been conducted and the results show that when node moving speed is relatively low, our routing protocol outperforms the other schemes such as Epidemic, Spray and Wait, FirstContact in delivery ratio and average hop count, while introducing an acceptable overhead ratio into the network. 
\end{abstract}


\section{Introduction}
\label{intro}

Delay
and Disruption Tolerant Networks (DTNs) has grown from relatively obscure research activities to a healthy research topic attracting both network designers and application developers\cite{Khabbaz2012}, due to that the communication model of the Internet is based on some inherent networking assumptions, e.g., the existence of a continuous end-to-end path between two nodes, the relatively short round-trip delays, the symmetric data rates and the low error rates\cite{Voyiatzis2012}.
However, in DTNs these assumptions usually fail, which leads the fact that the TCP/IP protocol does not work. 
Hence, many application protocols designed for the Internet architecture cannot operate well in the DTN scenarios such as Interplanetary Internet (IPN).
Besides, another major contributor to this trend is the observation that quite a few terrestrial networks exhibit delay-tolerant properties, albeit of different nature: from sparse mobile ad-hoc to sensor networks to mobile Internet access, it is found that delay tolerance also exists as an important element to describe communication behavior and to design protocols suitable for operation in the corresponding challenged networking environment\cite{Ott2012}.
Though there are many research achievements in Mobile Ad-hoc NETworks (MANETs), most of them assumed that the end-to-end
connection usually exists in the network. 

Besides, the research of communication has been extended into areas previously beyond the grasp of generic networking architectures in the past few years.
These networks have a variety of applications in situations that include crisis environments like emergency response and military battlefields\cite{Parikh2005}, vehicular communication\cite{Pereira2011}, mobile sensor networks\cite{Cha2013} and non-interactive Internet access in rural areas\cite{Demmer2007}. From the user’s perspective, not all the applications need a real-time response and thus such kind of requirements under the overlay networking should be relaxed with delay tolerance, so that the concept of MANETs could get much closer to reality.
And such kind of MANETs applications include e-mail service, 
news disseminating service and local temperature measure 
etc.
For these kind of application requests, there is an 
urgent need for designing a new networking architecture to 
converge all kinds of heterogeneous networks.
All these applications mentioned above address the importance of successful delivery instead of real-time response, and we call this characteristic delay-tolerance.
Kevin fall et.~al in \cite{Fall2003} present an early thinking about the subject and specifies the characteristics of networking in this kind of challenged networks.
The DTN architecture is defined in RFC 4838\cite{Cerf2007}.
and the Bundle Protocol (BP) is defined in RFC 5050\cite{Scott2007}, which runs between the application layer and the network layer. 

Most research achievements on DTNs focus on the design of routing protocols.
Though there are something common of routing between Internet and DTNs, routing in DTNs still faces many challenges.
Nodes in Mobile Delay Tolerant Networks (MDTNs) are more than fixed hosts and routers.
Furthermore, there are scarcely no pre-deployed infrastructures or assistant controlled nodes in delay tolerant MANETs, which means that there is no router-like device, and thus routing will be executed by all mobile nodes in the network cooperating with each other.
In other words, each node plays a role of router and hence acts in a “store-carry-forward” manner. 
Though that not relying on infrastructures or controlled nodes would highly increase the difficulty for routing, there are many reasons for paying more attention to the networks that own the ad hoc properties.
For example, for communication in military battlefields, a node statistically has a high possibility to be destroyed, which may lead to intermittent connections in the whole network.
Furthermore, the mobility pattern will change according to the tactical plan.
Thus nodes need to spontaneously form the network and then server as routers in order to deliver the message.
For vehicular to communicate, although we can pre-deploy some access points along the road or somewhere else at regular intervals, it is still considered to be much more expensive than assembling the wireless device on vehicle itself and relying on their mobility to communication.
And this has already been a hot research topic as a new branch of MANETs and named as Vehicular Ad hoc NETworks (VANETs).

Some proposed routing protocols rely on acquiring more than one-hop information or the assumption of the availability of global topology knowledge. 
Some others employ various kinds of message distribution protocols so as to obtain the needed information for routing. 
However, most of them are unrealistic to practically implement due to that partitioning, long delays, and dynamic topology may cause the failure of transmitting the needed routing information.
Taking a step back, although those information distribution mechanisms worked well, the collected information would be inaccurate or expired thus losing their real-time value for routing.

In this article, we intend to focus on investigating the mobile ad hoc networks with delay tolerance. Our work mainly differs from other achievements in the following aspects:
\begin{itemize}
\item The routing scheme is based on no more than one-hop information. There is no need to broadcast the link status to the whole network or to record any history informations for each node.
\item Most of the current research achievements devote to utilize a metric evaluating the relationship between the current node and the destination node (i.e. destination-aware metric). However, this implicitly requires that each node has the knowledge of all its prospective destination nodes and thus is of relatively low feasibility in DTNs. In this paper, we try to provide a kind of fresh thinking of the routing problem in DTNs, that let each node to ensure uniform geographical spread of the message copies.
\item It is necessary to cope with the overloads in the network due to that the multi-copy strategy usually introduces high message redundancy in the networks. As far as our information goes, this is the first paper to give the concept of message redundancy degree based on the message coverage in the network.
\end{itemize}
Our contributions are listed as following:
\begin{itemize}
\item
A distributed geographic routing scheme is proposed based on no more than one-hop information.
In details, a node only need its local one-hop neighbor(s)’ position information to make the next-hop choice.
The only assumption we made is that every node in the network is equipped with a local positioning device, thus having the ability to obtain its neighbors’ locations.
This kind of geographic routing is easy to implement in the real network, and would not trigger abundant calculations, so that the scarce energy resource could be saved.
\item
A simple criterion named as the Message Redundancy Degree (MRD) is proposed for measuring the message redundancy.
In this paper the accurate definition of 2-order MRD is given, and the concept of MRD can be extended to k-order $(k>2)$. 
We leave this to be our future work.
\item
A mechanism to cope with the message redundancy, by which we can reduce the size of useless message overlapping area during the whole routing process. 
We set a threshold value of 2-order MRD, for that when two nodes meet with a MRD value larger than the margin, the redundancy coping mechanism will be triggered for the purpose of saving the constrained buffer resource.
\item
The routing performance is evaluated and compared with other well-known routing schemes by extensive simulations. 
The simulation results show that when node moving speed is relatively low, our routing protocol outperforms the other schemes such as Epidemic, Spray and Wait, FirstContact in the aspect of delivery ratio and average hop count, while introducing an acceptable overhead ratio into the network.
\end{itemize}

The paper is organized as follows. 
In Section \ref{related} we report on previous works in the field of DTNs. 
In Section \ref{prelim} we introduce the preliminary and motivation of the routing strategy. 
In section \ref{keyprobs} we put forward key problems in routing design. 
Section \ref{routing} is devoted to mechanism routing design techniques. 
In section \ref{simu} we show the simulation results. 
Section \ref{con} eventually concludes this paper.

\section{Related Work}
\label{related}

Recent years have seen considerable research works proposed to address the issues of routing algorithms and protocols in DTNs. 
Most of them devote to deal with very sparse networks or intermittent connectivity is to reinforce connectivity on demand.
Some schemes bring additional communication infrastructures into the network, e.g., satellites, UAVs or controlled nodes, when necessary. 
The author in \cite{Zhao2004} propose to adopt Message Ferry (MF) under the sparse scenario where additional ferries are within the dedicated region to relay the message. 
As an extension based on MF, the work in \cite{Zhao2005a} focuses on using multiple ferries and designing their appropriate routes to maximize the throughput and minimize the delivery delay with four approaches proposed, which are SIngle Route Algorithm (SIRA), MUlti Route Algorithm (MURA), Node Relaying Algorithm (NRA) and Ferry Relaying Algorithm (FRA). 

Until currently, a set of congestion control mechanisms have been proposed in Deterministic DTNs, which is mainly implemented in the network with limited mobility or the static network with scheduled disruption interval. 
\cite{Cao2011} propose an active congestion control based routing algorithm that pushes the selected message before the congestion happens. 
\cite{Thompson2010} propose a novel node-based replication management algorithm which addresses buffer congestion by dynamically limiting the replication a node performs during each encounter. 
\cite{Dvir2010} use information about queue backlogs, random walk and data packet scheduling nodes to make packet routing and forwarding decisions without the notion of end-to-end routes. 
Furthermore, \cite{Ryu2010} proposes a two-level Back-Pressure with Source-Routing algorithm (BP+SR), which reduced the number of queues required at each node and reduced the size of the queues, thereby reducing the end-to-end delay.

To avoid the message flooding in the network, a straightforward approach is to replicate the message according to the utility metric rather than to replicate blindly. 
In \cite{Boloni2011} Turgut D. et.al. proposed the Bridge Protection Algorithm (BPA) which changes the behavior of a set of topologically important nodes in the network. These techniques protect the bridge node by letting some nodes take over some of the responsibilities of the sink. This method can significantly decrease the load of the nodes in the critical areas, while only minimally affecting the performance of the network.
\cite{Balasubramanian2010} treats DTN routing as a resource allocation problem that translates the routing metric into per-packet utilities that determine how packets should be replicated in the system. 
A random variable is used to represent the encounter between pairwise encountered nodes, thus nodes replicating messages in the descending order according to a marginal utility. 
\cite{Liu2012} presents two multi-copy forwarding protocols, called optimal opportunistic forwarding (OOF) and OOF-, which maximize the expected delivery rate and minimize the expected delay, respectively, while requiring that the number of forwarding operations of per message does not exceed a certain threshold. 
\cite{El-Azouzi2013} applies the evolutionary games to non-cooperative forwarding control in MDTNs, of which the main focus is on mechanisms to rule the participation of the relays to the delivery of messages in DTNs. 
In \cite{Martin-Campillo2013}, a utility function is introduced as the difference between the expected reward and the energy cost which is spent by the relay to sustain forwarding operations.

Besides, there are some inspired research achievements for transmissions in Mobile Sensor Networks (MSN). \cite{Ladislau2008} investigated the transmission scheduling problem for sensor networks with mobile sinks. The authors developed a dynamic programming-based optimal algorithm and described two decision theoretic algorithms which use only probabilistic models and do not require knowledge about their future mobility patterns. \cite{Turgut2009} describe and compare three practically implementable heuristic algorithms to control the transmission behavior of the nodes in the presence of mobile sinks. Besides, a graph-theory-based approach for calculating the optimal policy based on a complete knowledge is developed in \cite{Turgut2009}.

However, some of these works do not focus on the networks where nodes are mobile and simultaneously act in an ad hoc manner. 
In other word, there might be no available assistant controlled nodes or pre-deployed infrastructures for routing in these kinds of networks. 
Besides, most proposed routing protocols highly relay on  more than one-hop information. 
A few assume a predictable mobility pattern or a knowledge oracle set known in advance. 
And some others employ variety kind of information distribution mechanism to disseminate the needed knowledge about network topology. 
All these mentioned above may increase the difficulty to deploy those protocols in reality. 
Our work mainly differs from the research works above in two aspects.
First, we do not break up the ad hoc manner of networks, hence no controlled nodes or infrastructures being used. 
Second, we implement routing only based on one-hop available position information of neighbor nodes. The simulation results show that our algorithm have comparatively high delivery ratio and lowest hop count on average, furthermore having an acceptable overhead ratio.

\section{Preliminary and Motivation}
\label{prelim}

\subsection{Motivation}
\label{motiv}

Standard well-known distributed routing algorithms include distance vector (DV), path vector(PV) and link state (LS)\cite{Demmer2007}. 
In a routing algorithm using DV or PV, each node makes its next-hop decision by referring the recorded vectors obtained by the way that the node firstly collects its one-hop information and then assembles it with the vectors passed from others. 
However, a LS routing protocol assumes each node keeps track of the state of contacts and then distribute them to the network so as to make the whole network topology visible to each node. 
All the three above routing strategies usually need to obtain more than one-hop neighbor information, whereas on the one hand the network partitioning and high end-to-end delay may make the collected information out-dated and thus inaccurate, on the other hand , even the real time information can be obtained, the frequent disruptions in DTNs may incur the situation that updating messages flood in the whole network thus introducing high traffic loads.

To overcome the disadvantages of current routing schemes that make use of more than one-hop neighbor information, we propose a new routing scheme that relies only on the position information of its one-hop neighbor(s) for each node. 
On the first priority, the goal of our scheme is to maximize the delivery ratio.
Additionally, though optimizing the end-to-end delay is not of the vital emergency, high delivery ratio may in a way be benefit from quick delivery in DTNs, since nodes need not keep the copy of a delivered message thus saving the limited buffer resource and the restricted energy. 
Nevertheless, it is hard to find a measure which by taking would always strictly let both take a turn for the better. 
Luckily there is an intuition that shortening the distance between the current node holding the message and the destination node may consequently in some sense raise the message delivery ratio. 
Starting from this point, it enlightens us to design a routing scheme to spread the message to the direction of the destination as accurate as possible. 
We agree with the authors in \cite{Psaras2009} who observed that message duplication, on one hand, increases the delivery ratio and, on the other hand, decreases the delivery delay. 
However the naive replication strategy introduces unbearably high overheads into networks. 
So another task is to use duplication carefully in the trade-off between good performance and acceptable costs.
To achieve a comprehensive understanding of our goals for designing, we firstly list the key factors that could impact the performance and generality of routing, and then accordingly we make the assumptions of our network scene. 
The factors are listed as follows:
\begin{enumerate}
\item
The performance of some routing scheme highly depends on availability of the mobile pattern of nodes.
\item
Multi-copy strategy can add parallelism to routing by introducing more costs, which could accordingly cause high traffic loads, and consequently degrading the performance.
\item
Employing special controlled nodes can help routing by enhancing network connectivity at the cost of losing ad hoc network property.
\item
A practical routing algorithm in DTNs should work in a distributed online manner.
\end{enumerate}

With all the above considerations in mind, the basic assumption we made in this paper is that each node is equipped with a local positioning device. 
We do not use GPS for that a global positioning system may expose each user(s)’ individual position to the others. 
When the user does not want others to know its current position, it would not open the GPS service on its device. 
By referring the factors list above, we refine the most important principles of designing our routing algorithm respectively as follows:
\begin{enumerate}
\item
The routing algorithm makes no assumption of any available mobile pattern of nodes.
\item
The number of copies for each message is controlled by some means to avoid broadcast storming.
\item
No special controlled nodes or pre-deployed infrastructures are used in order to keep the ad hoc property of the whole network.
\item
No assumption of available global knowledge is made and thus our algorithm is more practical in realistic scenarios.
\end{enumerate}

\subsection{Mathematical Notations}
The format of mathematical notations is defined as following:
\begin{displaymath}
\mathcal{N}_{index}^{description}\left(par_1,par_2,\ldots \right)
\end{displaymath}
Therein $\mathcal{N}$ is a capital letter symbol, of which the superscript description is usually a simple description about the meaning of the notation. 
The optional subscript index is used only when $\mathcal{N}$ stands for a set, so we can use the subscript to identify each element of set $\mathcal{N}$. 
For example, each node in the network can be represented as $N_{EID}^{node}$ and hence $N^{node}$ represents the set of all nodes. 
$EID$ is an identifier for each node and is expressed simply as a word or a capital letter instead of the normal format defined in \cite{Cerf2007} for convenience in this paper. 
The parameters $(par1, par2, …)$ are also optional and used to represent other relative notation information. 
All the mathematical notations used in this paper are listed and explained in \tablename~\ref{math_table}.

\begin{table}
\centering
\caption{Mathematical notations}
\label{math_table}
\begin{tabular}{p{0.37\linewidth}<{\centering}||p{0.54\linewidth}<{\centering}}
\hline
\textbf{notation} & \textbf{meaning}\\
\hline \hline
$N_{X}^{node}$ & a node whose $EID$ is $X$ \\
\hline
$N^{node}$ & the set of all nodes \\
\hline
$P_{A}^{position}$ & the position of $N_{A}^{node}$ \\
\hline
$P^{neighbors}(A)$ & the positions set of $N_{A}^{node}$'s all neighbor nodes \\
\hline
$M_{m}^{bundle}$ & the message with $EID=m$ \\
\hline
$U_{m}^{utility}(N_{X}^{node},N_{A}^{node})$ &
the utility function to evaluate the quality of $N_{A}^{node}$ to $N_{X}^{node}$
for message $M_{m}^{bundle}$ \\
\hline
$V_{XA}^{vector}$ & vector determined by the two points $P_{X}^{position}$ 
and $P_{A}^{position}$ \\
\hline
$V_{\frac{\pi}{4}}^{vector}$ & vector determined by the two points 
$P_{X}^{position}$ and $P_{best}^{position}$ \\
\hline
$M^{delivered}$ & the number of delivered messages \\
\hline
$M^{relayed}$ & the number of relayed messages \\
\hline
$M^{created}$ & the number of created messages \\
\hline
$R$ & radius of node(s)' signal range \\
\hline
\end{tabular}
\end{table}

\section{Key problems in routing}
\label{keyprobs}
In this paper, we propose a Geographical Routing protocol based on ONE-hop information (GRONE). 
In this section, we give discussion and analysis about two key problems related to our routing algorithm, that are  how to choose relay nodes and how many relay nodes should be chosen.

According to principle 1 and principle 4 in \ref{motiv}, we have no idea about the position of the destination node, so the primary goal of our algorithm is to make messages spread uniformly to achieve the statistically highest possibility of reaching the destination node. 
Besides, we should try to ensure that messages spread as quickly as possible to achieve a short delivery latency. 
When using GRONE, all the routing information needed by a node A is listed as follows:
\begin{enumerate}
\item
The position of the $N_{A}^{node}$, denoted by $P_{A}^{position}$.
\item
The positions set of $N_{A}^{node}$'s all neighbor nodes, denoted by $P^{neighbors}(A)$.
\end{enumerate}

\subsection{Number of Replicas}
\label{num_of_relays}
When the network resources such as energy and buffer capacity are sufficient, using more relay nodes usually means higher delivery ratio and lower end-to-end delay. 
In fact, whenever message replicas are distributed among several nodes, if some of these nodes disappear(e.g. failures, destruction, power outage, alternation between sleep and active modes, etc.) the task of message delivery is delegated to the other remaining active nodes \cite{Khabbaz2012} and as a result the chance of a node encountering the destination raises with the increase in the number of carrier nodes. 
For example, once $N_{A}^{node}$ replicates the message $M_{m}^{bundle}$ to $N_{B}^{node}$,$N_{B}^{node}$ then is also responsible for delivering $M_{m}^{bundle}$. 
Next $N_{A}^{node}$ and $N_{B}^{node}$ respectively make their own routing decision so that two parallel paths to the destination node are generated. 
Nevertheless, network resources in DTNs are usually strictly constrained and limited and consequently routing in DTNs should take all the factors mentioned above into account. 
On one hand redundant messages occupy a lot of bandwidth and then be prone to cause network congestions. 
On the other hand processing or buffering messages consumes energy and nodes often choose to stop transmission when the energy is short \cite{Li2011} or to drop the message when no extra buffer space available. 
So it is necessary to limit the number of relay nodes in order to avoid dispensable redundancy \cite{Spyropoulos2008}.

\begin{figure*}
\centering
\subfigure[uniformly spreading\label{spread_uniform}]
{\includegraphics[width=0.3\linewidth]{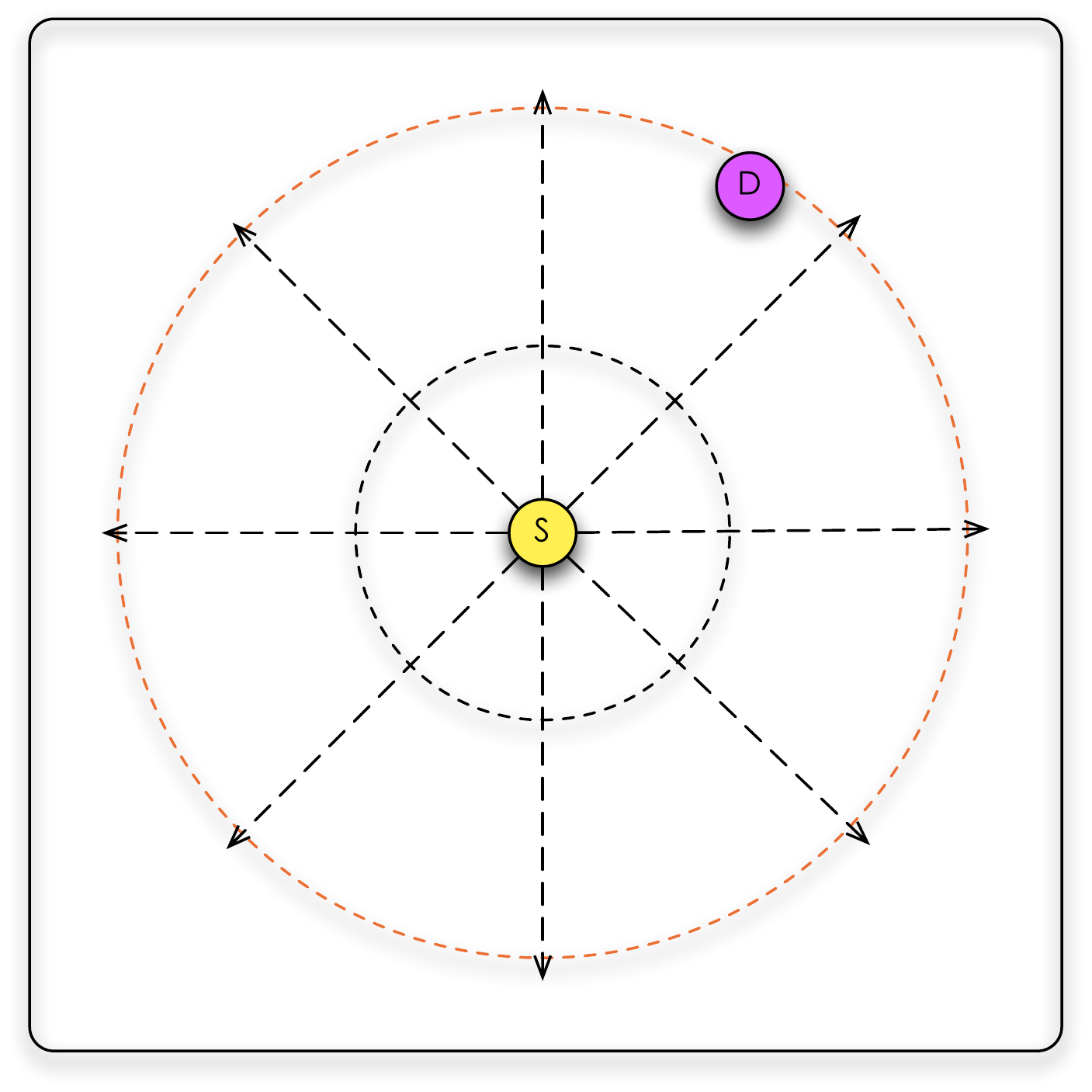}}
\hfil
\subfigure[nearest node to the source\label{nearest}]
{\includegraphics[width=0.3\linewidth]{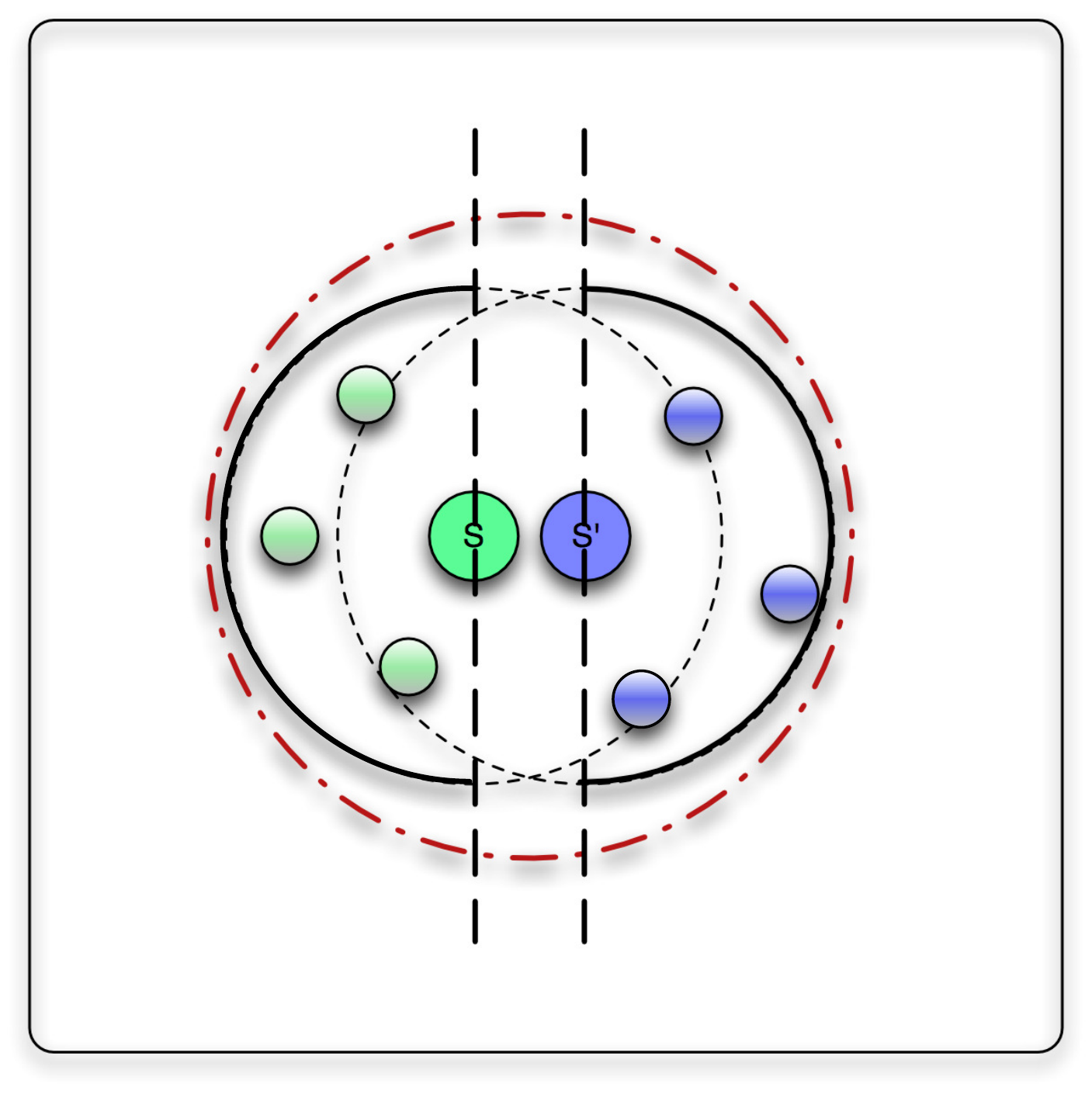}}
\hfil
\subfigure[farthest node to the source\label{farthest}]
{\includegraphics[width=0.3\linewidth]{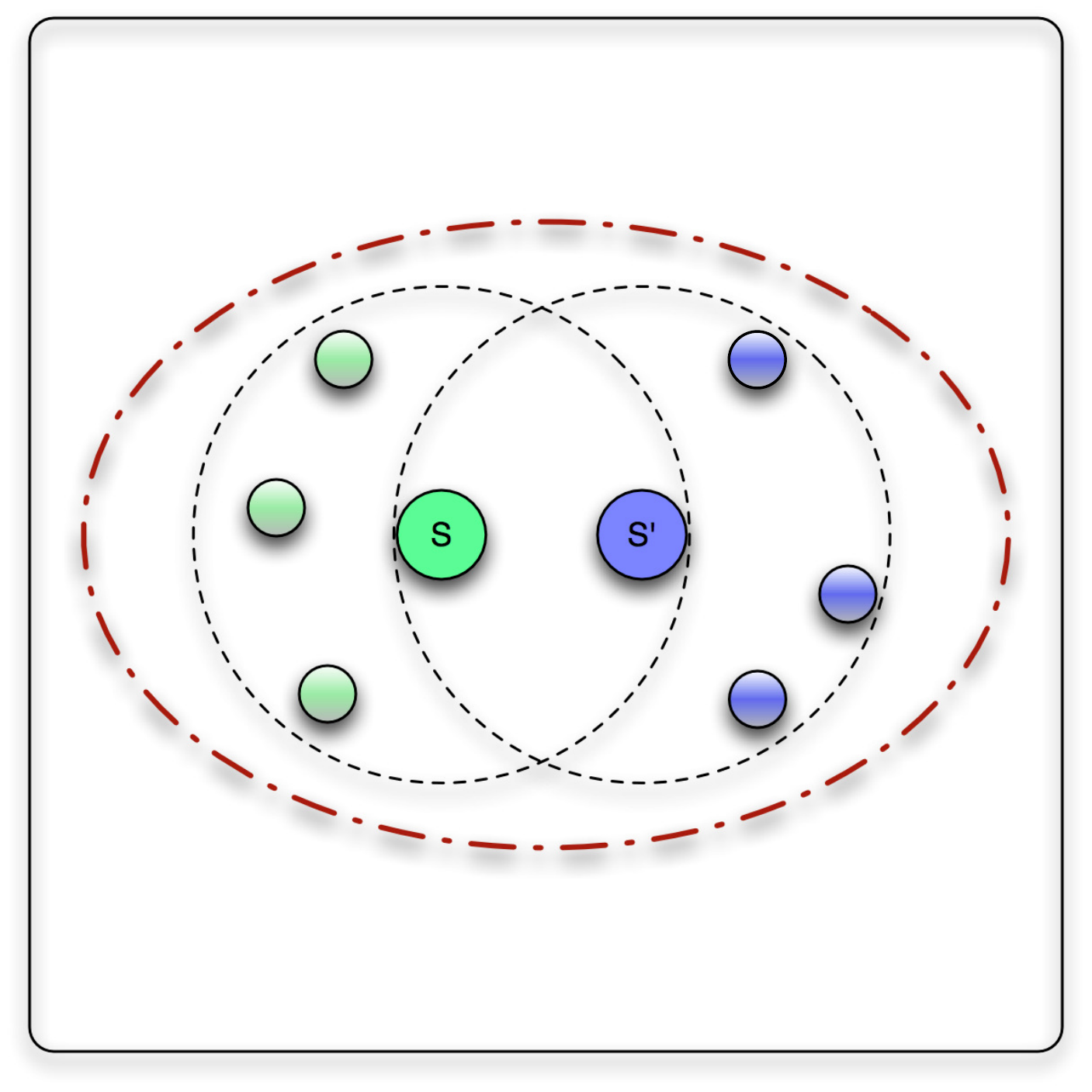}}
\caption{Appropriate number of relay nodes.}
\end{figure*}

To guarantee the message spread uniformly in a radiatory way, as shown in \figurename~\ref{spread_uniform}, one method is to let the source node pick four nodes as replication targets.
However, in the next step if $N_{C}^{node}$ still employs the same forwarding strategy as $N_{S}^{node}$ does, one of the forwarding direction of $N_{C}^{node}$ would be towards $N_{S}^{node}$ and then one copy of the message would be forwarded back to a position near the source $N_{S}^{node}$, which furthermore causes that messages are held by many nodes covering a roughly same area hence wasting the buffer resource. 
Besides, it is not easy for a node to choose as many as four next-hop relay nodes every time in a sparse mobile network. 
The key to cope with this problem is that we should let the message keep spreading uniformly while simultaneously not losing the consistency of the routing strategy.
Our method is to let the source node $N_{S}^{node}$ firstly pick the nearest node $N_{S'}^{node}$ and replicate messages to it, and then let both $N_{S}^{node}$ and $N_{S'}^{node}$ choose their own relay nodes in each semi-circle area, as illustrated in \figurename~\ref{nearest}. 
The reason why we choose the nearest instead of the farthest node is that though a farther node may increase the covering area of the message temporarily, the message covering area would approximately be a ellipse instead of a circle, as shown in \figurename~\ref{farthest}. 
In more detail, if the covered area is not uniform at the beginning, it would become more and more non-uniform consequently degrading the average performance of the algorithm. 
So we primarily try to keep the covered area more similar to a circle. 
Secondly, the source node, relay nodes should always keep spreading the message oppositely to the direction to the source node in order to increase the message covering area. 
In \figurename~\ref{choose_two}, $N_{S}^{node}$ is the source node of a certain message, and the large circle centered at S similarly represent the message covering area. 
$N_{S}^{node}$ is the node currently chosen to spread the message farther, and we make a line vertical to the line $SM$ and hence getting two sectors where $N_{A}^{node}$ and $N_{B}^{node}$ locate. 
If the routing can keep working as this manner, then we can guarantee that the message covering area would increase and finally cover the destination node, so the delivery task can be finished.

\subsection{Position of the Next Hop}
When the direction is determined, choosing a farther relay node may reduce the superfluous message covering area so as to efficiently utilize available network resources. 
As shown in \figurename~\ref{best_pos}, assuming that $N_{X}^{node}$ is the current node to make the routing choice, we can find it intuitively that $N_{best}^{node}$ would be the best relay node, since it is the farthest node to $N_{X}^{node}$. $N_{worst}^{node}$ locates on the same position of $N_{X}^{node}$ and we regard this as the worst choice since replicating the message to $N_{worst}^{node}$ does not make any increase of the message covering area.

\begin{figure}
\centering
\includegraphics[width=0.5\linewidth]{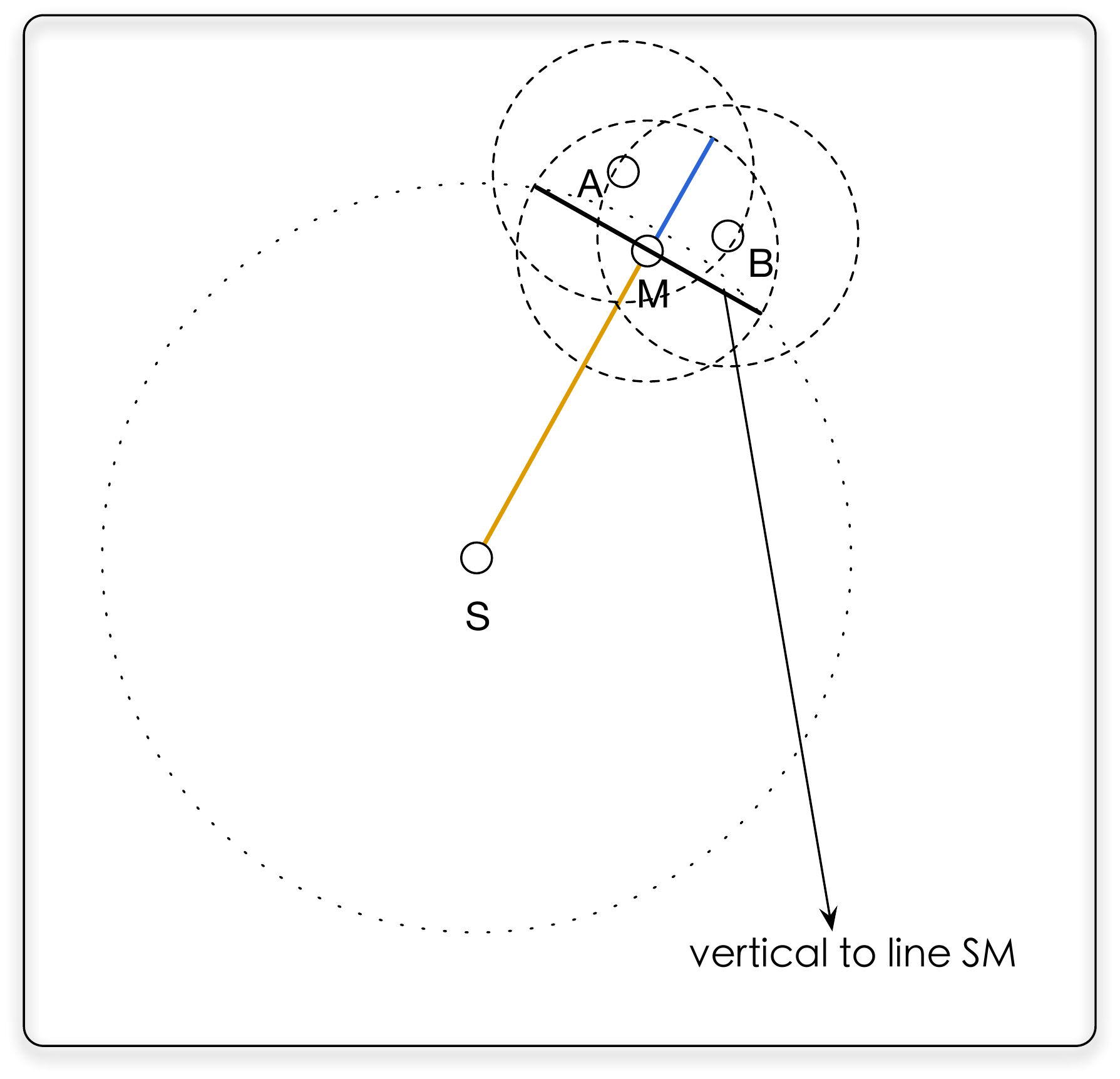}
\caption{Choose two nodes to achieve a radiatory routing.}
\label{choose_two}
\end{figure}

\begin{figure}
\centering
\includegraphics[width=0.7\linewidth]{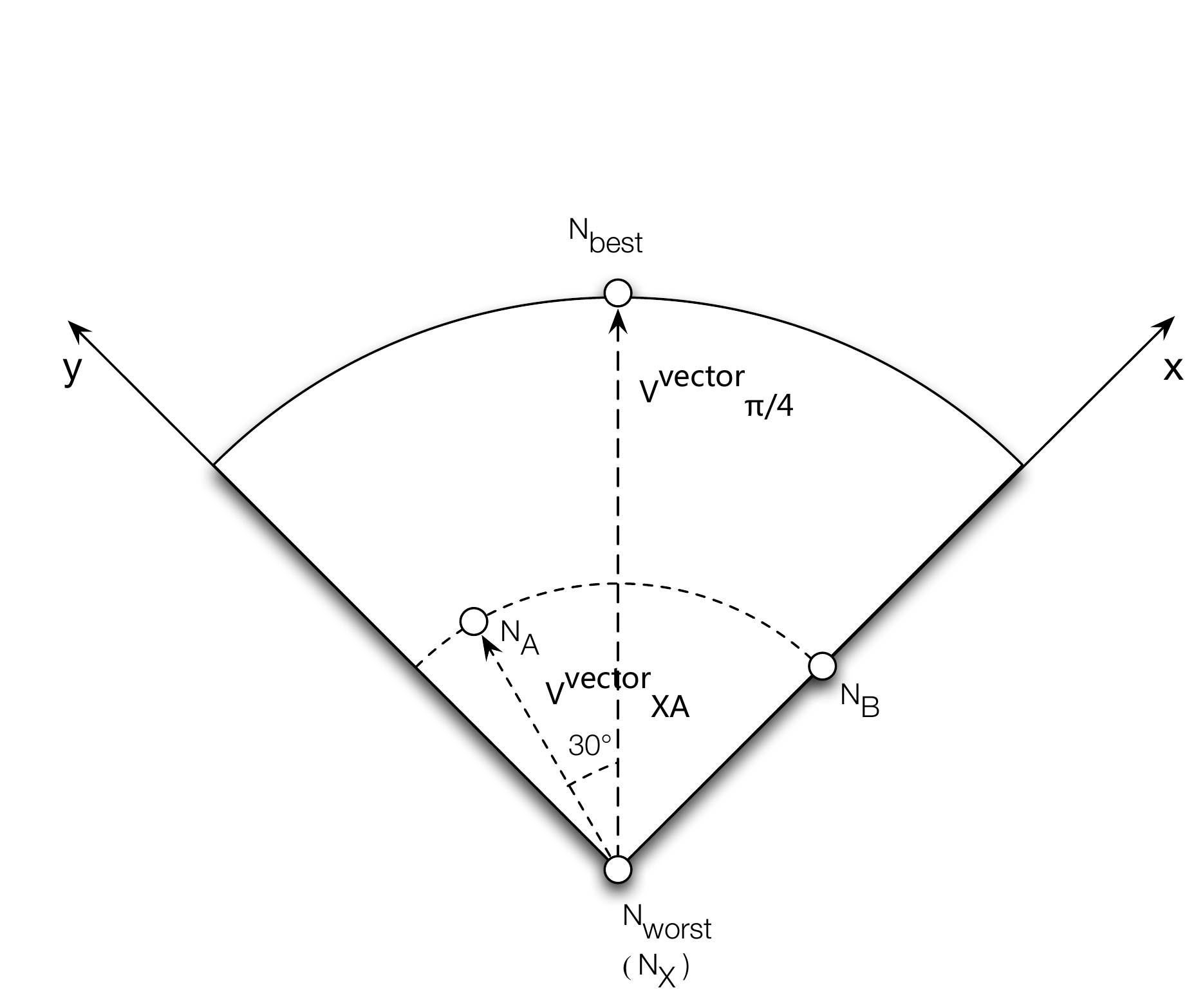}
\caption{The best and the worst position of relay nodes.}
\label{best_pos}
\end{figure}

The next question is that which is the most suitable direction. 
Since each node makes its own decision and have no idea about any of its surrounding nodes’ choices, a good way to guarantee the routing algorithm perform well from the statistical perspective is to always make a best choice based on estimating its surrounding nodes’ expected choices. 
In the next section we will design a utility function by which we can achieve this goal and quantify the suitability of each promising relay node.

\section{Routing design}
\label{routing}

\subsection{Utility Function for Choosing the Relay Node}
Since we have no information about the destination node, e.g. the distance or the direction, the only way to maximize the possibility of contacting the destination is to spray messages as uniformly as possible. 
For conveniently explaining our solution, we firstly define the name “peer node”:
\begin{definition}[peer node]
\label{peer node}
For a certain node $N_{A}^{node}$, its peer node refers to the node also chosen as the relay by the direct predecessor node of $N_{A}^{node}$. 
\end{definition}
As shown in \figurename~\ref{choose_two}, $N_{B}^{node}$ is a peer node of $N_{A}^{node}$ and vice versa, since that both $N_{A}^{node}$ and $N_{B}^{node}$ have been chosen as relay nodes by their predecessor node $N_{M}^{node}$. In order to maintain a uniform spraying, each node’s next-hop choice should also refer to the choices of its peer(s). However, a node could not rely on a “query mechanism” to obtain its peer’s choices, otherwise the principle P4 would be violated, since if that we could not guarantee the node and its peer to be in the reliable transmission range of each other. Conversely, the query is not a “one-hop” operation, thus violating the principle.

The only way to obtain the peer node’s option is to let each node “guess” the direction that its peer is most likely to choose. 
However, the possible relay node position is determined by its possibility to show up in each position of the sector. 
Thus our task is to statistically get the most possible direction for the prospective relay nodes to show up. 
And from the statistical significance it is known that an expectation of the direction is the best choice and leads to a minimum variance. 
And consequently we aim to get the expected shown up direction for a node in the sector. 
All the node(s)’ position are represented as set $P^{position}$ and the cooperation system is built as shown in \figurename~\ref{best_pos}. 
Because each node has the same probability of showing up in any position of the sector, $P_{X}^{position}$ obeys a uniform distribution, and thus we have the probability density function, as shown in equation \ref{equ1}.
\begin{equation}
\label{equ1}
f(x,y)=\frac{1}{S_{sector}}=\frac{1}{\frac{1}{4}\pi R^2}
\end{equation}
And we can get the expected direction by equation \ref{equ2}.
\begin{equation}
\label{equ2}
\begin{aligned}
E[\theta]&=\int\!\!\!\int_{S}f_{X,Y}(x,y)\theta dS \\
&=\int_0^{\frac{\pi}{2}}\!\!\!\int_0^{R}\theta\cdot\frac{1}{\frac{1}{4}\pi R^2}\cdot
rdrd\theta \\
&=\left.\frac{1}{2}\theta^{2}\right|_{0}^{\frac{\pi}{2}}\cdot\frac{1}{\frac{1}{4}\pi R^2}\left.
\cdot\frac{1}{2}r^2\right|_{0}^{R}\\
&=\frac{\pi}{4}
\end{aligned}
\end{equation}

Thus the wise way to keep messages spray uniformly is to let each node always choose the node nearest to the expected direction, as calculated above, $\pi\left/4\right.$. 
Like most routing schemes in DTNs, GRONE is also based on the predefined utility describing the appropriateness of the forwarding operation for a specific message. 
We define the utility function as equation \ref{equ3}.
\begin{equation}
\label{equ3}
\begin{aligned}
&U_{m}^{utility}\left.\right(N_{X}^{node},N_{A}^{node}\left)\right.\\
=&\frac{1}{2R}D^{distance}\left.\right(N_{X}^{node},N_{A}^{node}\left)\right. \\
&+\left.\right(1+\frac{\sqrt{2}}{2}\left)\right.
\left.\right(\cos\left.\right<V_{XA}^{vector},V_{\frac{\pi}{4}}^{vector}\left>\right.-
\frac{\sqrt{2}}{2}\left)\right.
\end{aligned}
\end{equation}
Besides, for $\forall~N_{Y}^{node}$, we let
\begin{displaymath}
\cos\left.\right<V_{X,X}^{vector},V_{X,Y}^{vector}\left>\right.=\frac{\sqrt{2}}{2}.
\end{displaymath}

As shown in \figurename~\ref{best_pos}, we denote the direction from $N_{X}^{node}$ to $N_{A}^{node}$ and the direction of the bisector of the sector by $V_{XA}^{vector}$ and $V_{\frac{\pi}{4}}^{vector}$ respectively. 
Then we could explain the meaning of equation \ref{equ3}. 
The first part
\begin{displaymath}
\frac{1}{2R}D^{distance}\left.\right(N_{X}^{node},N_{A}^{node})
\end{displaymath}
is the criteria for evaluating the distance.
The rest of equation \ref{equ3}
\begin{displaymath}
\left.\right(1+\frac{\sqrt{2}}{2}\left)\right.
\left.\right(\cos\left.\right<V_{XA}^{vector},V_{\frac{\pi}{4}}^{vector}\left>\right.-
\frac{\sqrt{2}}{2}\left)\right.
\end{displaymath}
describes the degree of approximation between $V_{XA}^{vector}$ and the
expected optimal direction $V_{\frac{\pi}{4}}{4}$ discussed above.
Thus the range of either part is $\left.\right[0,1/2\left]\right.$, and then
we have $U_{m}^{utility}\in[0,1]$.
Investigating the two nodes $N_{best}^{node}$ and $N_{worst}^{node}$ in 
\figurename~\ref{best_pos}, the direction of $N_{best}^{node}$ is $V_{\frac{\pi}{4}}^{vector}$
and the distance between it and $N_{X}^{node}$ is $R$.
Since we can get the cosine value of the angle formed by these two vectors
\begin{displaymath}
\cos\left.\right<V_{X,best}^{vector},V_{\frac{\pi}{4}^{vector}}^{vector}\left>\right.=1
\end{displaymath}
we then obtain the highest utility value from equation \ref{equ3}:
\begin{displaymath}
\begin{aligned}
&U_{m}^{utility}\left.\right(N_{X}^{node},N_{best}^{node}\left)\right. \\
=&\frac{1}{2R}\cdot R+\left.\right(1+\frac{\sqrt{2}}{2}\left)\right.
\left(\right.1-\frac{\sqrt{2}}{2})=1
\end{aligned}
\end{displaymath}

Meanwhile, node $N_{worst}^{node}$ locates on the same position as $N_{X}^{node}$,
and thus we have
\begin{displaymath}
D^{distance}\left.\right(N_{X}^{node},N_{worst}^{node}\left)\right.=0
\end{displaymath}
and
\begin{displaymath}
\cos\left.\right<V_{X,best}^{vector},V_{\frac{\pi}{4}^{vector}}^{vector}\left>\right.=\frac{\sqrt{2}}{2}
\end{displaymath}
Under this case, we get the lowest utility as following:
\begin{displaymath}
\begin{aligned}
&U_{m}^{utility}\left.\right(N_{X}^{node},N_{worst}^{node}\left)\right.\\
=&\frac{1}{2R}\cdot 0+(1+\frac{\sqrt{2}}{2})(\frac{\sqrt{2}}{2}-\frac{\sqrt{2}}{2})=0 
\end{aligned}
\end{displaymath}
We can get $U_{m}^{utility}(X,A)$ and $U_{m}^{utility}(X,B)$ in the same way:
\begin{displaymath}
\begin{aligned}
&U_{m}^{utility}\left.\right(N_{X}^{node},N_{A}^{node}\left)\right. \\
=&\frac{1}{2R}\cdot\frac{R}{2}+(1+\frac{\sqrt{2}}{2})(\frac{\sqrt{3}}{2}
-\frac{\sqrt{2}}{2}) \\
\approx &0.5213 \\
&U_{m}^{utility}(N_{X}^{node},N_{B}^{node}) \\
=&\frac{1}{2R}\cdot\frac{R}{2}+(1+\frac{\sqrt{2}}{2})(\frac{\sqrt{2}}{2}
-\frac{\sqrt{2}}{2}) \\
=&0.5 
\end{aligned}
\end{displaymath}
Thus we have the sequence of the choice priority:
\begin{displaymath}
N_{best}^{node}>N_{A}^{node}>N_{B}^{node}>N_{worst}^{node}
\end{displaymath}

Equation \ref{equ3} guarantees that when locating on the same direction, a farther prospective relay node leads to a larger utility value. 
Meanwhile Equation \ref{equ3} ensures that when the distances are same, a direction being more approximate to the expectation has a higher utility value. 
Thus the value of $U_{m}^{utility}$ is determined by both direction and distance of the relay node, which are taken into consideration when making routing decisions so that the position information of neighbor nodes could be efficiently utilized.

\subsection{Coping with the Message Redundancy}
We represent the set of all messages in the network as $M^{bundle}$, and when $N_{A}^{node}$ generates $M_{m}^{bundle}$, any position within the reliable signal range of $N_{A}^{node}$ can be viewed as a “covered point” and $N_{A}^{node}$ should be responsible for the transmission task in the covered area. 
Let us image that $N_{B}^{node}$ is also holding a copy of $M_{m}^{bundle}$ and the distance between $N_{A}^{node}$ and $N_{B}^{node}$ is less than $R$, as shown in \figurename~\ref{mrd}. 
\begin{figure}
\centering
\includegraphics[width=0.6\linewidth]{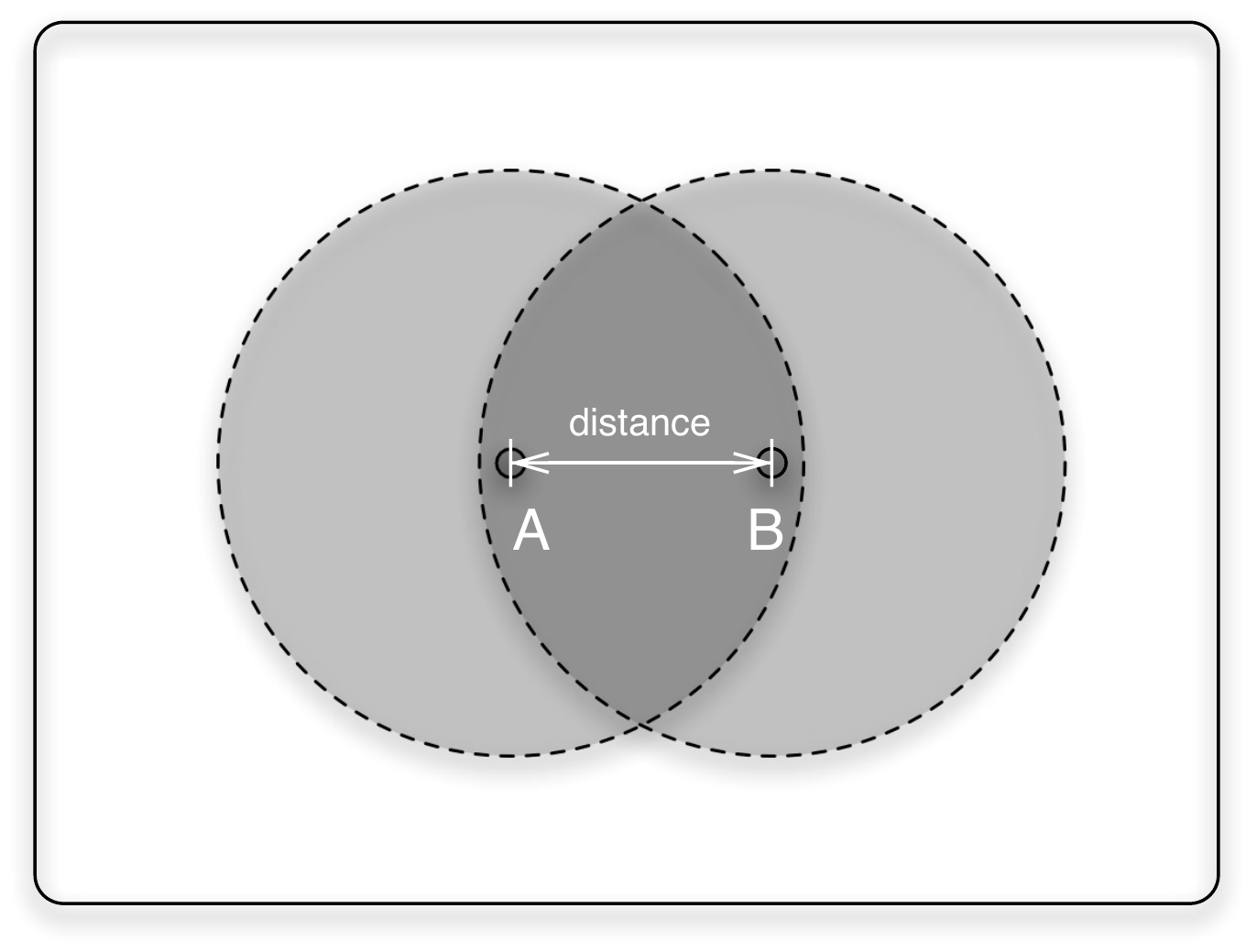}
\caption{Overlapping area of message $M_{m}^{bundle}$}
\label{mrd}
\end{figure}
In this case there is an overlapped area where both $N_{A}^{node}$ and $N_{B}^{node}$ have the ability to finish the transmission task for $M_{m}^{bundle}$. 
If the overlapped area is very large, then there is no need for both $N_{A}^{node}$ and $N_{B}^{node}$ to hold the same copy of $M_{m}^{bundle}$. 
In other words, either $A$ or $B$ takes charge of a similarly same area for transmission and thus wasting the network resource. 
To describe the mentioned above redundancy, we define a variable, named as Message Redundancy Degree (MRD), as following:
\begin{definition}[Message Redundancy Degree (MRD)]
Assuming that $N^{overlap}$ is a set containing one or more than one nodes and all the nodes within it hold the message $M_{m}^{bundle}$, we denote the \emph{k-order Message Redundancy Degree (k-MRD)} of $N^{overlap}$ for $M_{m}^{bundle}$ by $S_{m}^{overlap}\left.\right(k,N^{overlap}\left)\right.$, and we have
\begin{displaymath}
\begin{aligned}
&S_{m}^{overlap}\left.\right(k,N^{overlap}\left)\right. \\
=&\textnormal{the size of area covered by}\ k\ \textnormal{node(s)' in}\ N^{overlap}
\end{aligned}
\end{displaymath}
\textnormal{where} $k\leq\left.\right|N^{overlap}\left|\right.$.
\end{definition}

In \figurename~\ref{mrd}, the dark grey area is the 2-order MRD of $N_{A}^{node}$，$N_{B}^{node}$ for $M_{m}^{bundle}$, represented as $S_{m}^{overlap}$. Then the relationship of $S_{m}^{overlap}$ and \[d=D^{distance}(N_{A}^{node},N_{B}^{node})\] can be derived, for simplicity we let
\begin{displaymath}
S(d)=S_{m}^{overlap}\left.\right(2,\left.\right\{N_{A}^{node},N_{B}^{node}\left\}\right.\left)\right.
\end{displaymath}
thus having
\begin{displaymath}
\begin{aligned}
S(d)&=S_{m}^{overlap}\left.\right(2,\left.\right\{N_{A}^{node},N_{B}^{node}\left\}\right.\left)\right.\\
&=2R^{2}\left.\right|\arccos \frac{d}{2R}\left|\right.
-dR(1-\frac{d^2}{4R^2})
\end{aligned}
\end{displaymath}
and when $d=R\left/\right.2$, we have $S(R/2)≈2.167R^2$
The overlapping area is approximately 70\% of a reliable transmission area, and we denote this value as 2-Margin. 
When there exists a 2-MRD greater than 2-Margin between any two nodes in the network, we delete a redundant copy of the same message from either of them so as to save the limited network resources and lower the traffic loads. 
The following theorem declares the necessity of this strategy for coping the redundancy.
\begin{theorem}
If GRONE keeps running by always choosing the best node selected by function $U^{utility}$, then at some moment in the future there would be two nodes of which the 2-MRD is larger than 2-Margin.
\end{theorem}
\begin{proof}
Since the function $S(d)$ is monotone decreasing, 
we only need to prove that we can find two nodes 
between which the distance is less than $R/2$ at 
some moment in the future. 
As shown in \figurename~\ref{proof} 
if $\angle MSN=2\theta$ and $\angle ASN=\frac{1}{2}\angle MSN=\theta$, 
then $\angle MNR=\theta$. Additionally since $QR$ is vertical 
to $SN$, we have $\angle RNN'=\pi/4$. 
Thus since $\angle MNN'=\angle MNR+\angle RNN'$, 
we have $\angle MNN'=\theta+\pi/4$, 
where $\theta\in[0,\pi/4]$. 
Consequently we can derive from the isosceles trapezoid 
$MNN'M'$ to derive the formulation of the relationship 
between $|N'M'|$ and $\theta$, as following
\begin{displaymath}
|N'M'|=\sqrt{2}R-2R\cos(\theta+\pi/4)
\end{displaymath}
By running GRONE, the distance between ``future 
$N_{A}^{node}$'' and $N_{S}^{node}$ will be 
larger and larger. 
Since $N_{M}^{node}$ and $N_{N}^{node}$ both locate 
on best position evaluated by $U^{utility}$,  
$|MN|$ is a constant value $\sqrt{2}R$. 
So when $|SA|$ becomes larger, $\theta$ would decrease, 
which hence shortens $|N'M|$. 
Because $\theta\in(0,\pi/4]$, we have $N'M\in(0,\sqrt{2}R]$. 
Thus there would be a pair of very close nodes between 
which the distance is less than R/2.
\end{proof}
\begin{figure*}
\centering
\subfigure[$A$'s two relay nodes $M$ and $N$]
{\includegraphics[width=0.37\linewidth]{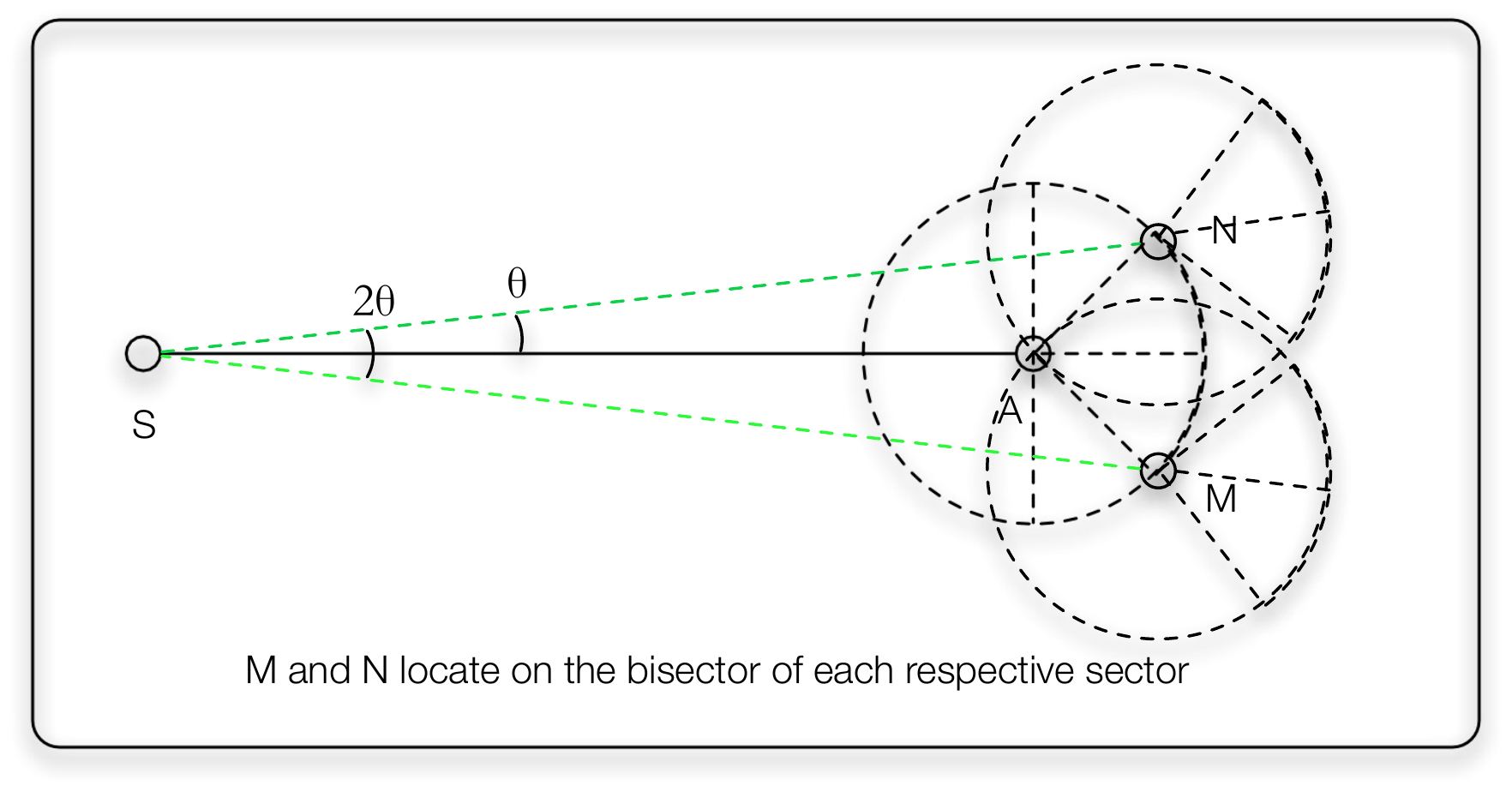}}
\hfil
\subfigure[$QR$ is vertical to $SN$]
{\includegraphics[width=0.37\linewidth]{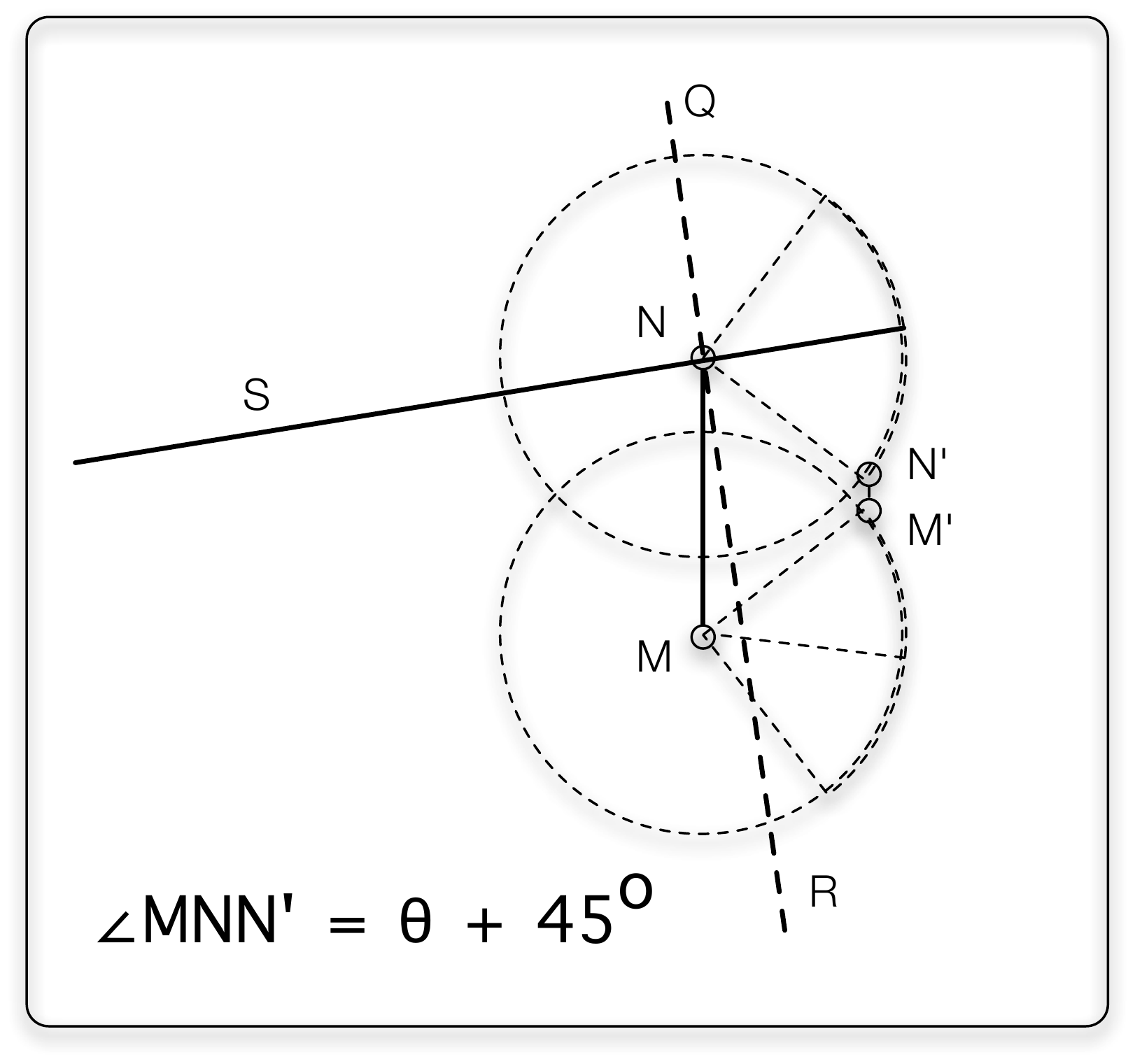}}
\caption{Get the isosceles trapezoid $NMM'N'$.}
\label{proof}
\end{figure*}

Now we have finished the discussion about the critical 
issues in our routing design, relay nodes selection and 
dealing with redundancy. In the next section we will 
describe our routing protocol mechanisms.

\subsection{Routing Protocol}
Since GRONE does not rely on any global or more than 
one hop information for each node, we do not need any 
information distribution schemes such as DV or LS and 
thus only let each node broadcast a “Hello” message 
when it joins the network. 
So the non-data traffic for 
each node will only be relative to the node density in 
the network.
A control “Hello” message 
includes some information about its originated node e.g., 
$EID$, position and a summary vector of all the messages 
held in it\cite{Vahdat2000}. 
For each node, all the neighbors within 
its transmission range will receive the “Hello” message 
and then add the corresponding position and $EID$ 
information to their neighbors table. Besides, we let 
the data message contain its source node position, for 
the convenience of being referred by other relay nodes.

\figurename~\ref{replication} describes two replication 
strategies in GRONE, 
named naive replication and utility-based replication 
respectively. Each node will first check its messages 
one by one, and transmit all the deliverable ones. 
Then for the remaining messages the node will employ 
either of the two replication strategies for custody 
transmission. Assuming that the current node making 
the replication decision is $N_{X}^{node}$, $N_{X}^{node}$ will 
check its neighbors table firstly and then utilize 
the naive replication strategy when there is only 
one active neighbor. Whereas when there are more 
than one neighboring nodes exist, $N_{X}^{node}$ will make 
use of the utility based replication method thus 
avoiding the message flooding in the network. 
\begin{figure}
\centering
\includegraphics[width=0.7\linewidth]{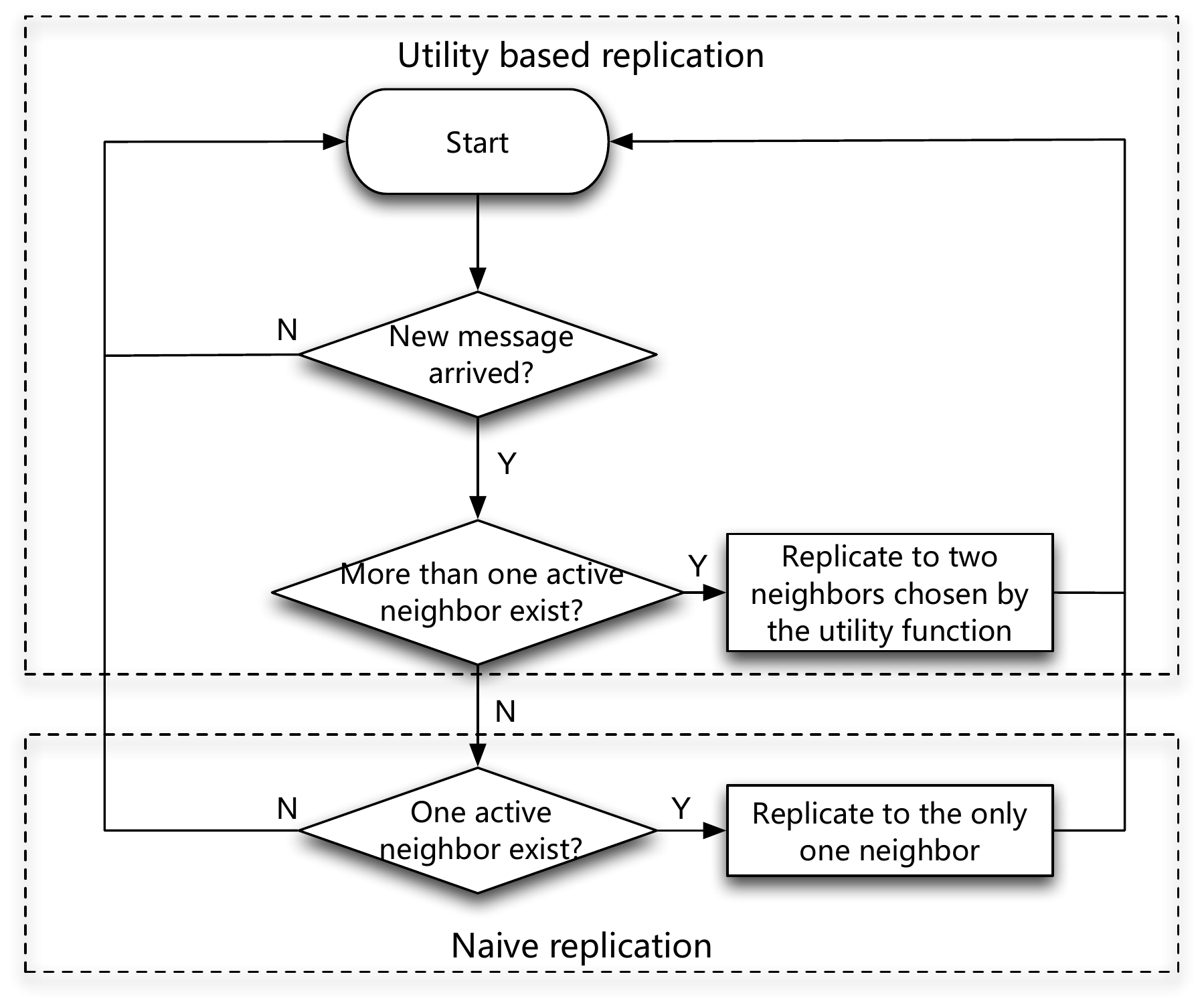}
\caption{Naive replication and utility based replication.}
\label{replication}
\end{figure}

\figurename~\ref{hello} shows the table updating and 
the message redundancy coping mechanism, both of which 
are based on exchanges of the ``Hello'' message among nodes.
\begin{figure}
\centering
\includegraphics[width=0.7\linewidth]{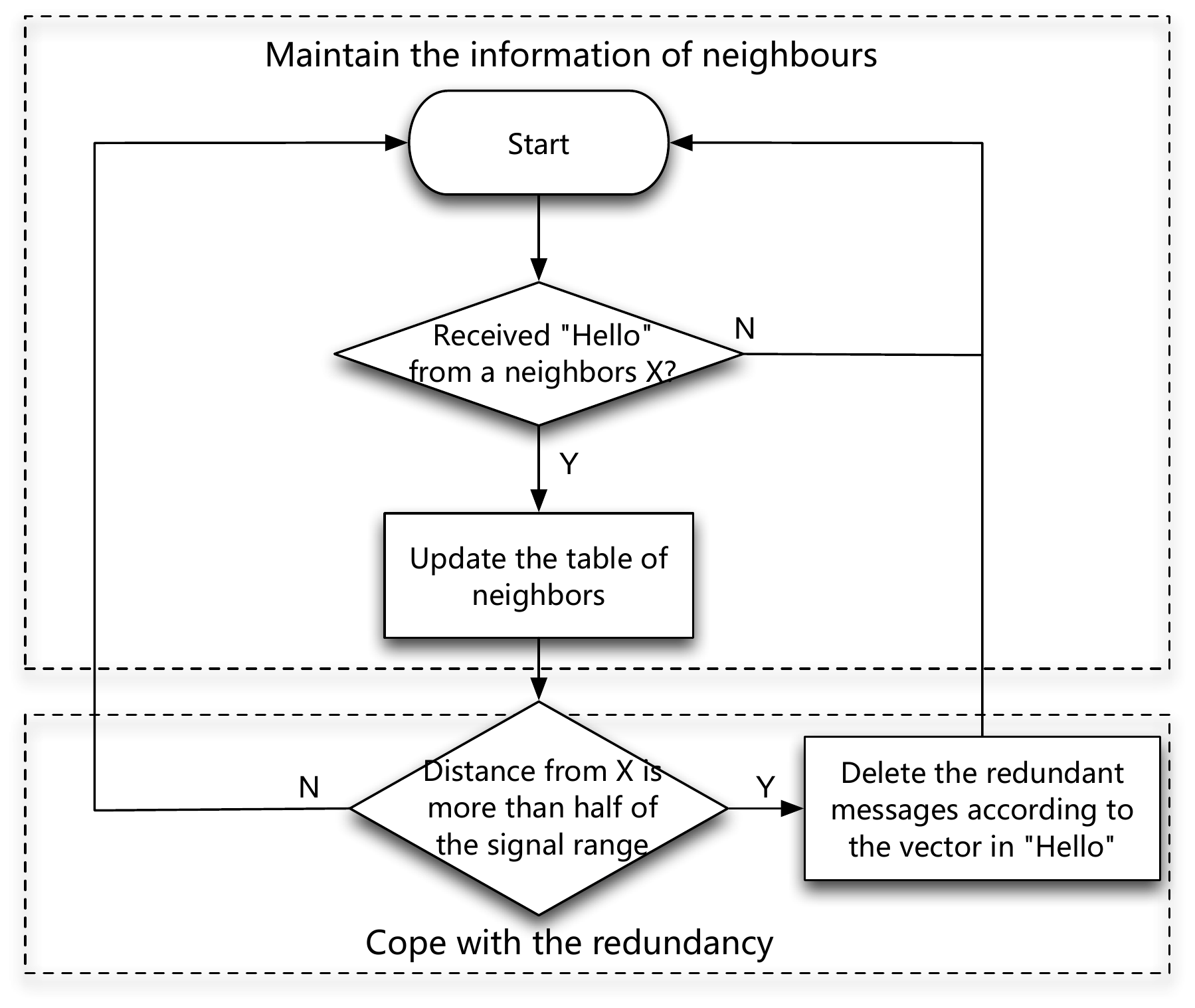}
\caption{Cope with redundancy by "Hello" messages when two nodes are close.}
\label{hello}
\end{figure} 
Since the energy is limited in delay tolerant MANETs 
and there is a tradeoff between the listening frequency 
and energy consumption of sending probing ``Hello'' 
messages\cite{Wang2009}, we should take full advantage 
of each ``Hello'' message, and thus we bind the message 
redundancy coping mechanism with the ``Hello'' message. 
Here comes an interesting question. 
In the discussion mentioned above, we know that the 
redundancy coping mechanism will delete the message 
from either of the two nodes. 
So, should the message still be transmitted to a node 
whose distance from $N_{X}^{node}$ is less then $R/2$? 
In GRONE, the answer is yes. 
One reason is that we intend to split the redundancy 
coping mechanism from the replication task for the 
purpose of making GRONE simple and clear. 
Furthermore, a ``now bad node'' might be a ``prospective 
future good relay'', since that node might move to a 
farther place where no node can cover now. 
Thus if we did not replicate the message to the node 
under this circumstance, we would in some sense miss 
a chance to deliver the message further. 
Moreover, even the “now bad node” is also ``a future bad node'', 
it will still be handled by the redundancy coping mechanism, 
since every two arrived ``Hello'' message will trigger it, 
as shown in \figurename~\ref{hello}. 
In this case, the node will check the summary vector contained 
in the ``Hello'' message with its own, and delete the redundant 
messages so as to save the buffer resources.

The ``Hello'' message can aid a node in notifying its neighbor(s)’ 
attendance, and thus we let each node maintain its neighbors table 
by exchanging ``Hello'' messages. 
However a neighbor can be considered available only when it is active, 
which means that we should also have a notification scheme for nodes 
to delete the invalid neighbors. 
To handle this, each node deletes the corresponding neighbor entry 
when its expected ``Hello'' message has been missed more than two times. 
The behind reason is that one time ``Hello'' missing may caused by 
some reason such as MAC conflicting or unforeseen delay etc. 
Though we can revive the deleted neighbor entry to the table when a new 
``Hello'' message arrives, operations on the table will take some time 
and bring energy consumption. 
Consequently a fault-tolerant mechanism for ``Hello'' message mentioned 
above is necessary for relieving this problem.

\section{Simulation}
\label{simu}

\subsection{Simulation Environment and Settings}
The results are evaluated by the ONE \cite{Keranen2009} simulator, 
which is widely used for the evaluation of DTN protocols 
and is actively developed at the Department of 
Communications and Networking at the Aalto University. 
ONE is a time-stepped simulator, 
which means that the accuracy of a simulation result is 
based on a configuration specific clock-step. 
A lower clock step usually leads to a faster simulation process. 
In most cases of our simulations, the clock-step is set 
to $0.1$ seconds, while in some extreme compute-intensive 
cases we choose to slightly increase the clock-step to obtain 
an acceptable simulation time. And to compensate the 
inaccuracy introduced by lowering the simulating 
particle size, each simulation is executed five times. 
By averaging, we eventually give the final results.

\begin{table}
\centering
\caption{Simulation settings}
\label{simu_set}
\begin{tabular}{
p{0.45\linewidth}<{\centering}||
p{0.25\linewidth}<{\centering}||
p{0.15\linewidth}<{\centering}}
\hline
\textbf{parameter name} & \textbf{default} & \textbf{range} \\
\hline \hline
number of nodes & 120 & - \\
world size($m^2$) & 1000 & - \\
tickets in Binary S\& W & 18 & - \\
message TTL(min) & 20 & - \\
simulation time(hours) & 5 & - \\
message size(KB) & 500 & - \\
node buffer size(MB) & 6 & 2--10 \\
transmission range(m) & 100 & 20--180 \\
node moving speed(m/s) & 0.5 & 0.2--0.8 \\
movement model & Random Walk & - \\
message interval(s) & 40 & 20--60 \\
transmission speed(KBps) & 250 & - \\
\hline
\end{tabular}
\end{table}

Again, we emphasize that, in our network situation, the 
global or more than one-hop information is difficult to obtain. 
Thus we choose to compare GRONE with three other routing 
protocols, namely Epidemic, Binary Spray \& Wait and FirstContact 
(FC), for the reason that all these routing protocols are 
independent of any global knowledge assumption or any 
information distribution mechanism. 
These three algorithms can be classified into 
three different categories. Epidemic makes use of network 
resources as much as possible to achieve the 
highest delivery ratio, while its mad flooding introduces 
high level of overhead ratio into the network. 
When the network resource is sufficient, Epidemic usually 
does the trick. 
Binary Spray \& Wait \cite{Spyropoulos2005} can be viewed as a flooding routing 
limited in hop count and number of message copies. 
Finally, FirstContact \cite{Jain2004} is the simplest single copy routing 
protocol, which lets each node always forward messages to the first encountered neighbor.
By comparing with all these routing protocols, we analyze 
the feature and performance of GRONE. 
The settings of the simulation scenario are listed in 
\tablename~\ref{simu_set}. 
\figurename~\ref{ONE} shows one of the simulation scenario snapshots. 
\begin{figure}[tbp]
\centering
\includegraphics[width=0.95\linewidth]{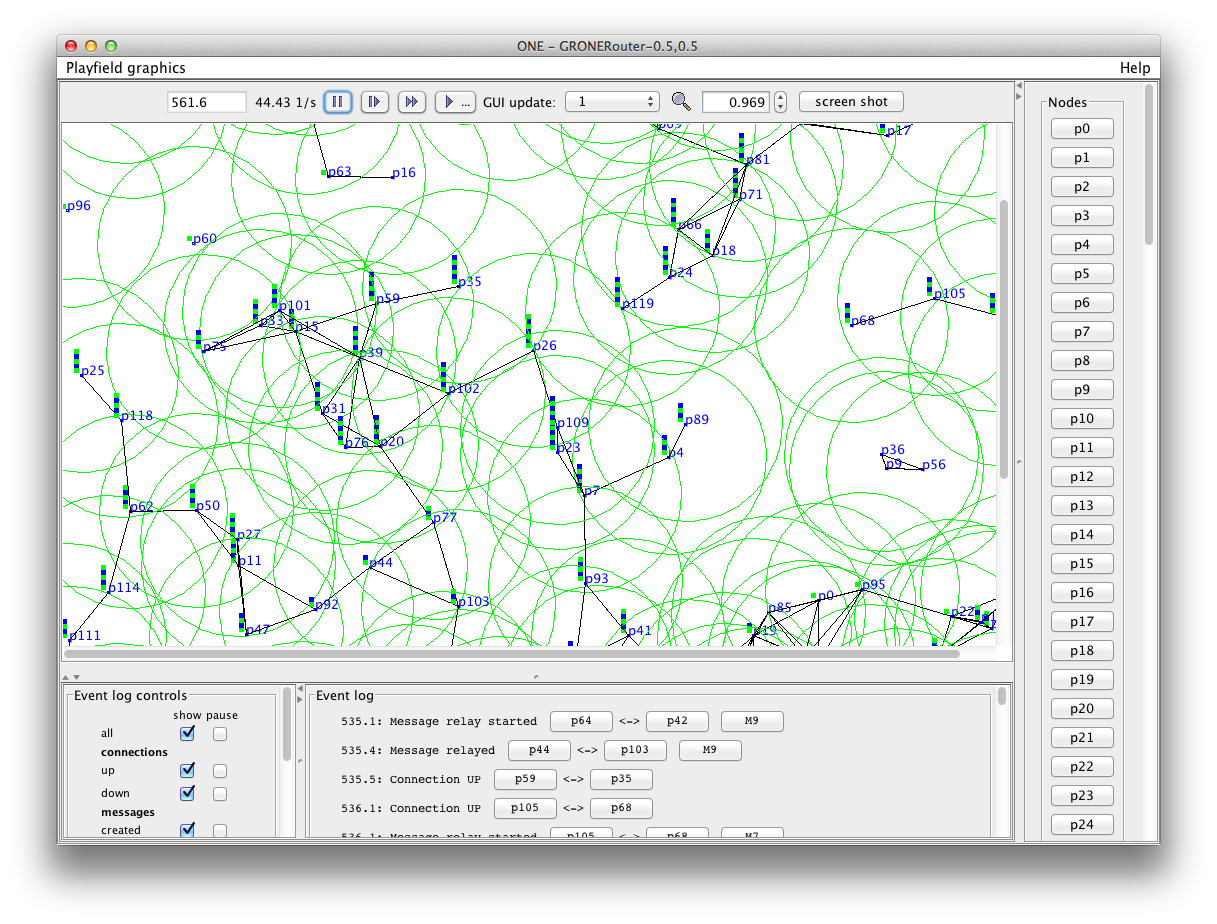}
\caption{The simulation scene in ONE simulator.}
\label{ONE}
\end{figure}
We have the nodes deployed randomly for each simulation 
and make all the nodes move according to the Random Walk 
movement model. 
Particularly, we set the node moving speed to be a 
relatively 
low value, for the purpose of investigating GRONE's 
performance with less assistance supported by node mobility. 
We find that GRONE outperforms the other three routing 
algorithms when the moving speed of nodes is relatively low, 
which indicates that a geographic routing strategy might be 
a good solution for the network scenario where node position 
is comparatively stable.

The evaluative criteria employed in our simulation are defined as following:

\[
\left\{
\begin{array}{ll}
message~delivery~ratio & = \frac{M^{delivered}}{M^{created}} \\
average~hop~count & = \frac{\sum total\ hops\ of\ every\ bundle}{M^{created}} \\
overhead~ratio & =\frac{M^{relayed}-M^{delivered}}{M^{delivered}}
\end{array}\right.
\]


We implement our evaluation by varying message generated 
interval, node buffer size, node moving speed and node 
transmission range. 
We firstly evaluate the performance 
of four algorithms under the three criteria mentioned above. 
Secondly we evaluate the stability of GRONE by changing 
the simulation time and the number of all nodes. 
Thirdly we investigate the factors that might affect 
the performance of GRONE by adjusting the node moving 
speed and node transmission range. 
All the simulation results are analyzed in subsection \ref{simu_analysis}.

\subsection{Simulation Result and Analysis}
\label{simu_analysis}
\figurename~\ref{delivery} shows the delivery ratio of the four routing protocols. 
\begin{figure*}[!tbp]
\centering
\subfigure[varying interval\label{delivery_vs_interval}]
{\includegraphics[width=0.48\textwidth]{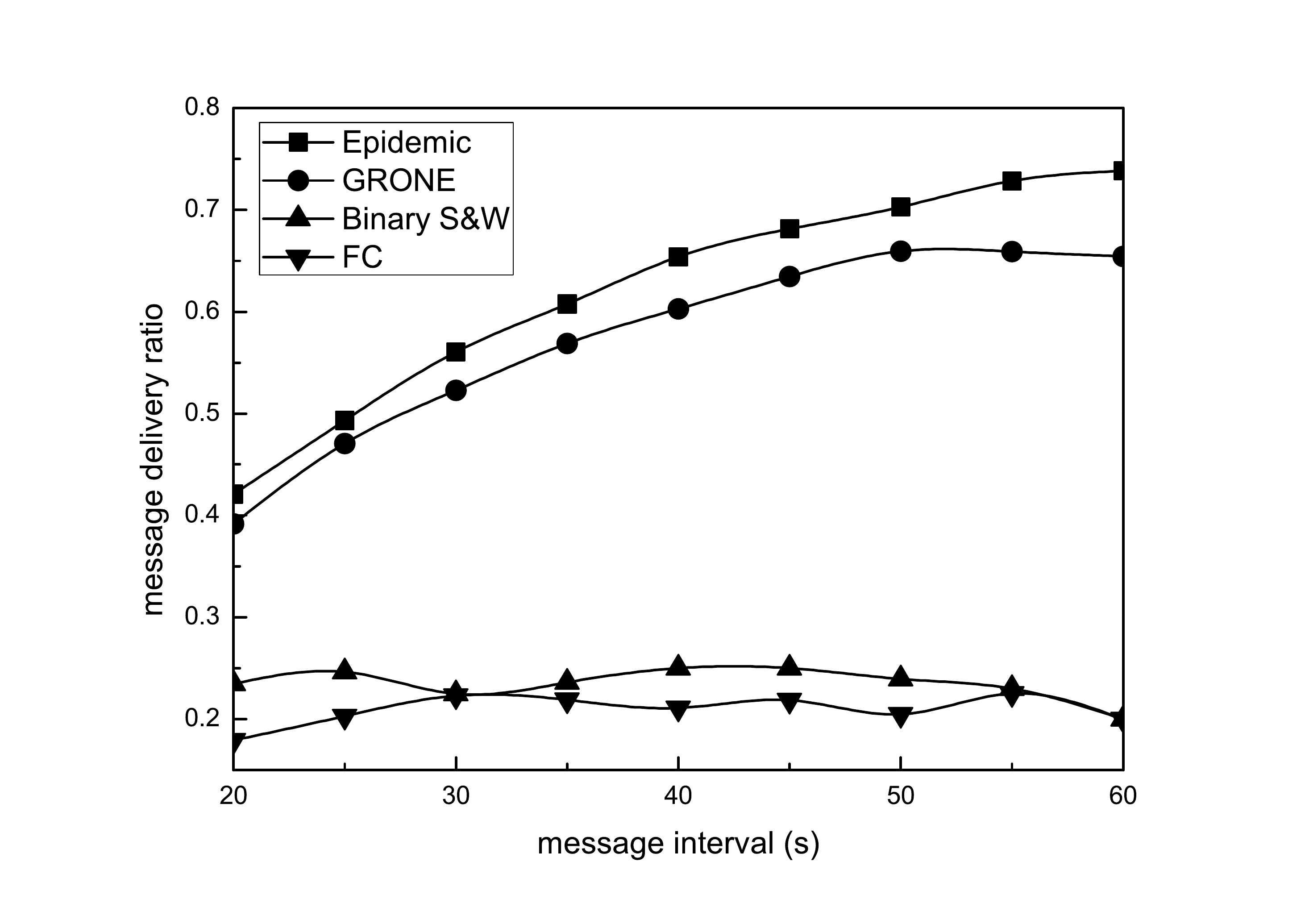}}
\subfigure[varying buffer\label{delivery_vs_buffer}]
{\includegraphics[width=0.48\textwidth]{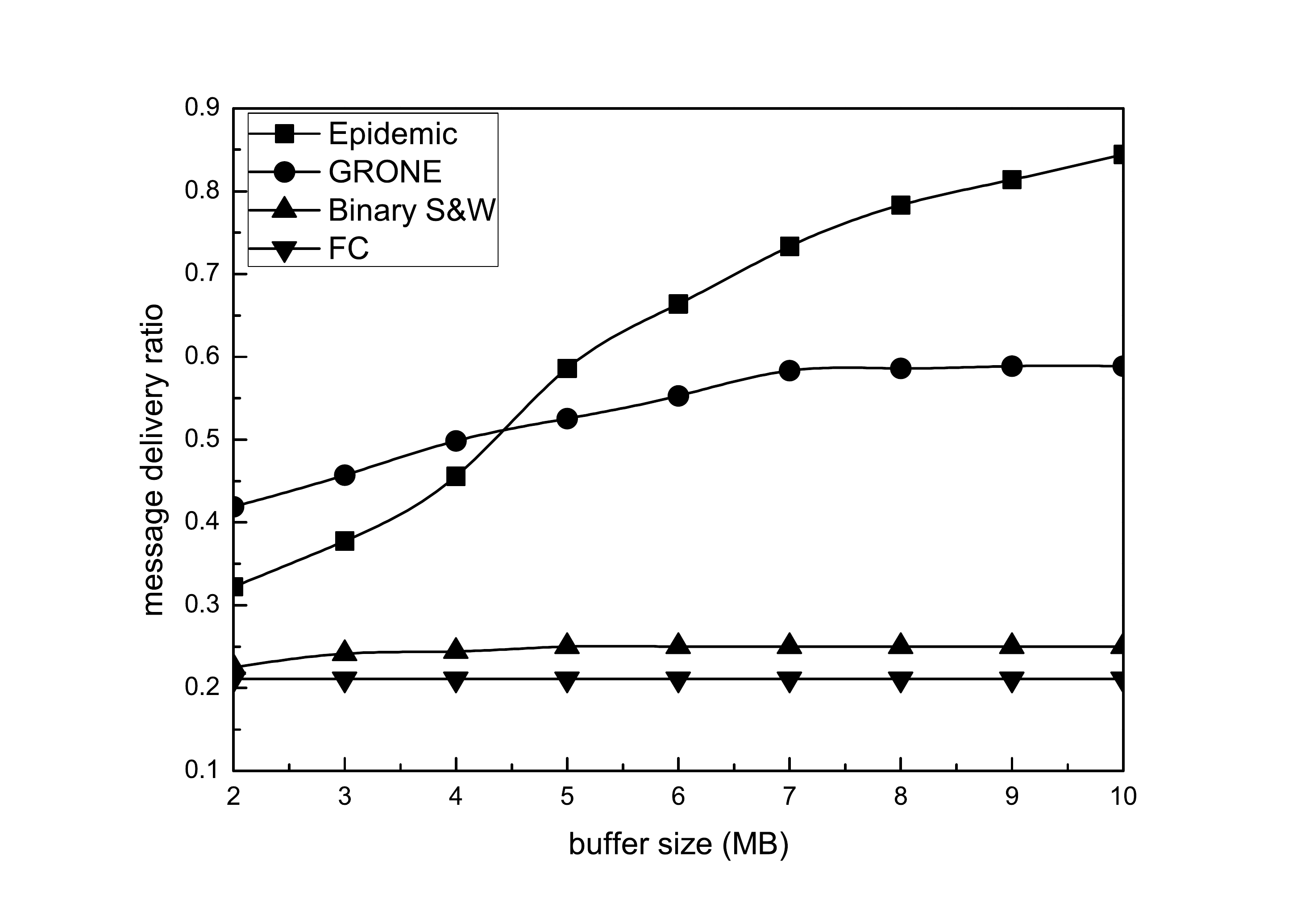}}
\\
\subfigure[varying speed\label{delivery_vs_speed}]
{\includegraphics[width=0.48\textwidth]{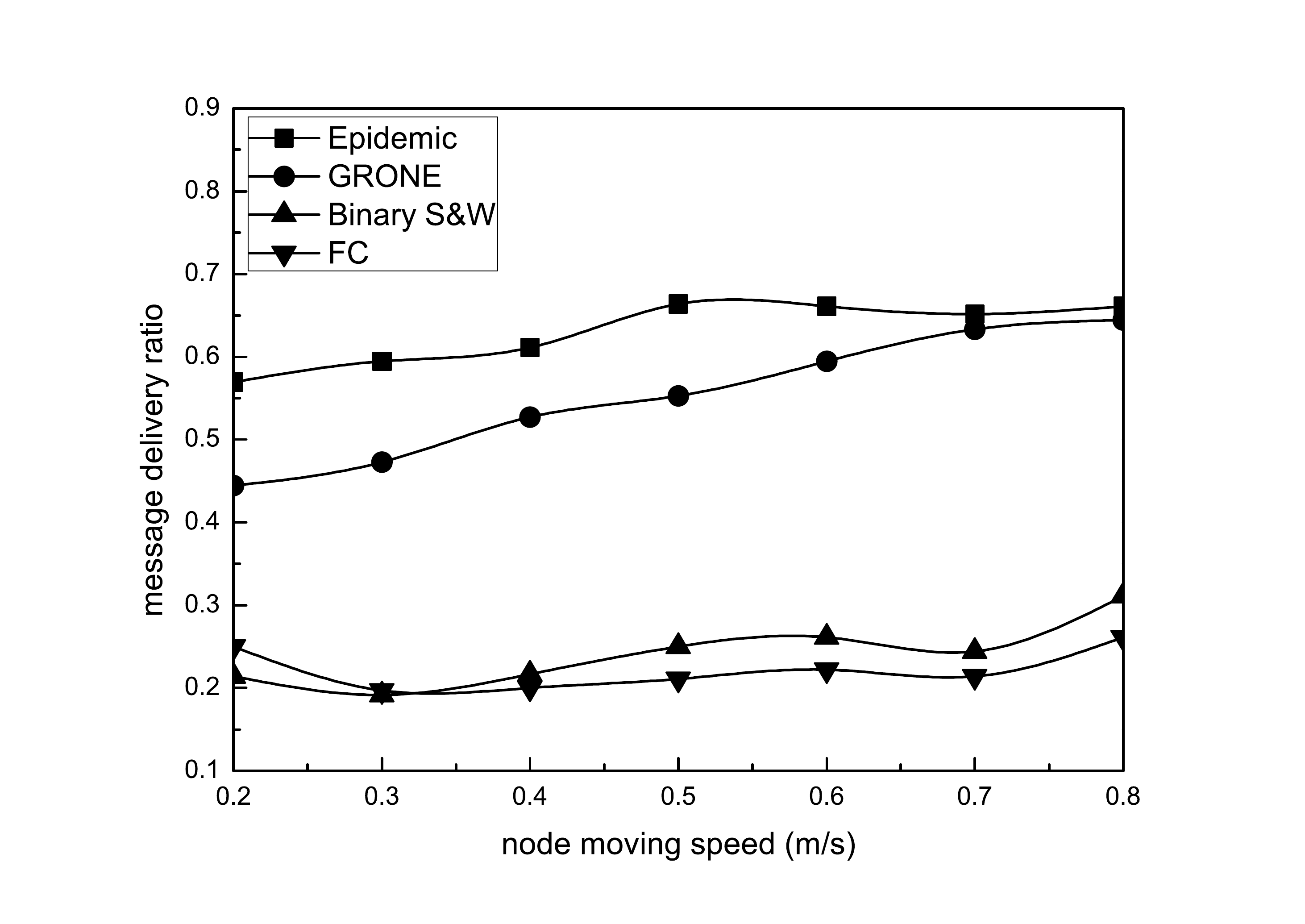}}
\subfigure[varying range\label{delivery_vs_range}]
{\includegraphics[width=0.48\textwidth]{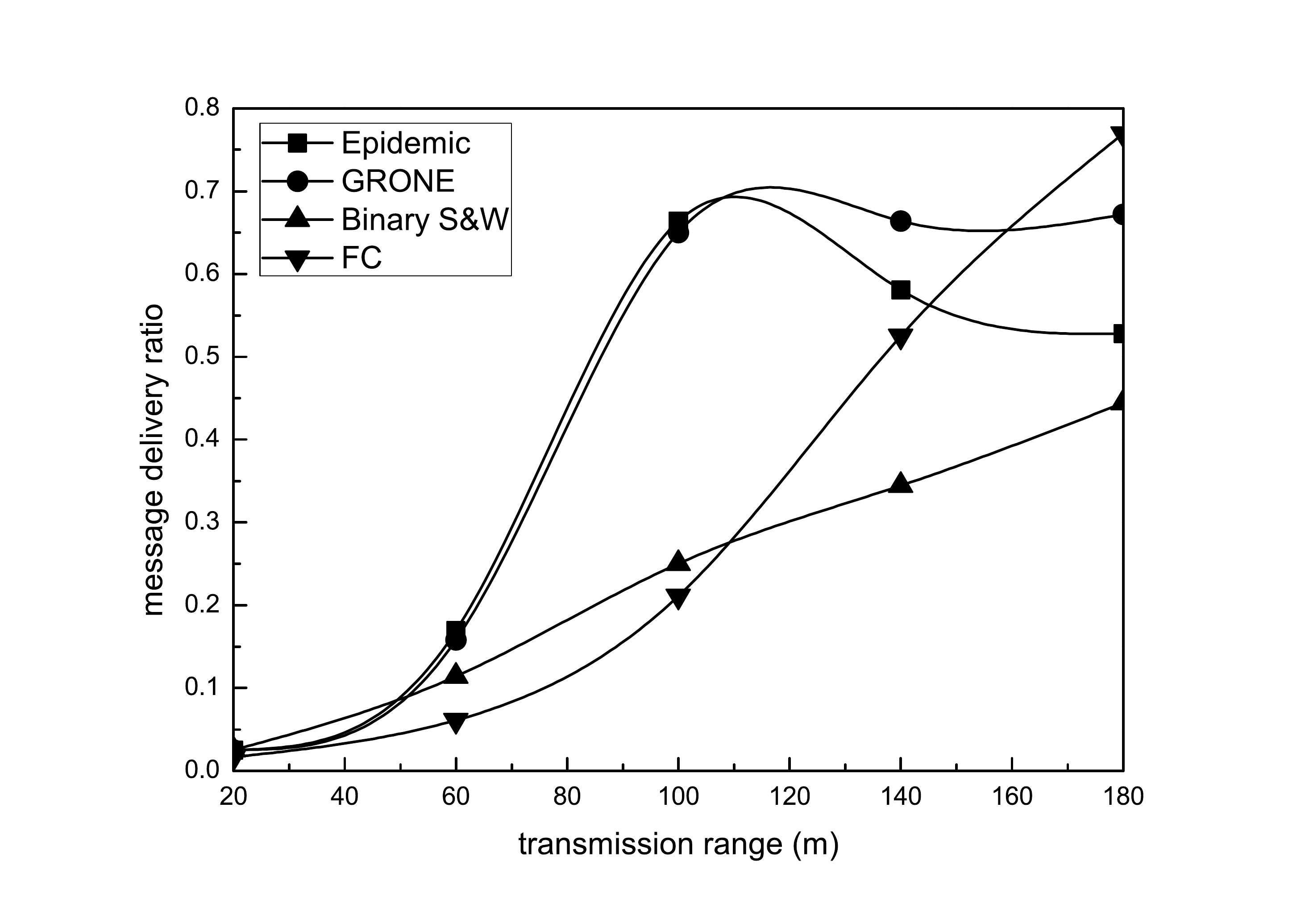}}
\caption{Message delivery ratio vs message interval, buffer size, node moving
speed and transmission range.}
\label{delivery}
\end{figure*}
In \figurename~\ref{delivery_vs_interval}, all the four routing 
protocols have a relatively low delivery ratio when the 
message interval is short, since 
that there are a great number of generated messages which
fully fill the node(s)' buffer, thus leading a high message
dropping probability. 
With the increment of message interval, the message 
delivery ratio of both Epidemic and GRONE arises. 
GRONE almost has the same high delivery ratio as Epidemic, 
while Binary Spray \& Wait and FC show worse performance 
than the first two protocols. 
Observing the two curves of 
Epidemic and GRONE,
we find that GRONE performs almost as well as Epidemic by 
resorting its intrinsic utility based replication strategy. 
And the results reflect that the redundancy coping mechanism only degrades a little for message delivery probability. 

As shown in \figurename~\ref{delivery_vs_buffer}, Epidemic has the highest delivery ratio when the
buffer size is sufficient, as expected. 
When the buffer size is about 4M,
the message delivery ratio of GRONE is appropriately equal 
to Epidemic, and both of them perform better with the available buffer 
resource increasing. 
However, when the buffer resource is strictly constrained, the 
delivery ratio of GRONE is higher than that of Epidemic, 
for the reason that it employs a utility based replication strategy instead 
of a naive one. Furthermore, GRONE has a redundancy coping mechanism, 
thus saving the limited buffer resource.
Binary Spray \& Wait and FC still have a non-ideal performance. 

\figurename~\ref{delivery_vs_speed} shows that the delivery ratio of all the four routing algorithms slightly arise with the increase of node moving speed. 
When node moving speed is relatively low, Binary Spray \&  Wait and FC do not have an ideal performance. Especially, FC employs the single-copy strategy and thus forwards message in a ``non-concurrent way'', so it has the worst performance. Beyond that, we know from the simulation 
results that multi-copy strategy can be regarded as an efficient method for increasing the delivery probability. 

In \figurename~\ref{delivery_vs_range}, when the transmission range is short, all the four algorithms have an approximately same low delivery ratio, due to that there are too few available contacts to be used in the network. With the increase of the transmission range, the connectivity of the whole network continues enhancing, and the delivery ratio of both Epidemic and GRONE arises much more quickly than Binary Spray \& Wait and FC. The reason is that both of
them generate sufficient enough message copies and make full
use of the network resource. Binary Spray \& Wait achieves an acceptable performance only when the connectivity of the network is quite strong. However, in that case the single copy routing FC performs best in terms of least resource
consumption. From the analysis mentioned above, it is found that multi-copy strategy is quite helpful for enhancing the message delivery probability when the node moving speed is low. 
Besides, GRONE roughly has the same high delivery ratio as Epidemic. Furthermore, GRONE is a better choice than Epidemic when the buffer resource is constrained.

In \figurename~\ref{hop}, we evaluate the average hop count 
performance of the four algorithms. A lower average hop count usually means a more efficient relay operation. 
\begin{figure*}
\centering
\subfigure[varying interval\label{hop_vs_interval}]
{\includegraphics[width=0.48\textwidth]{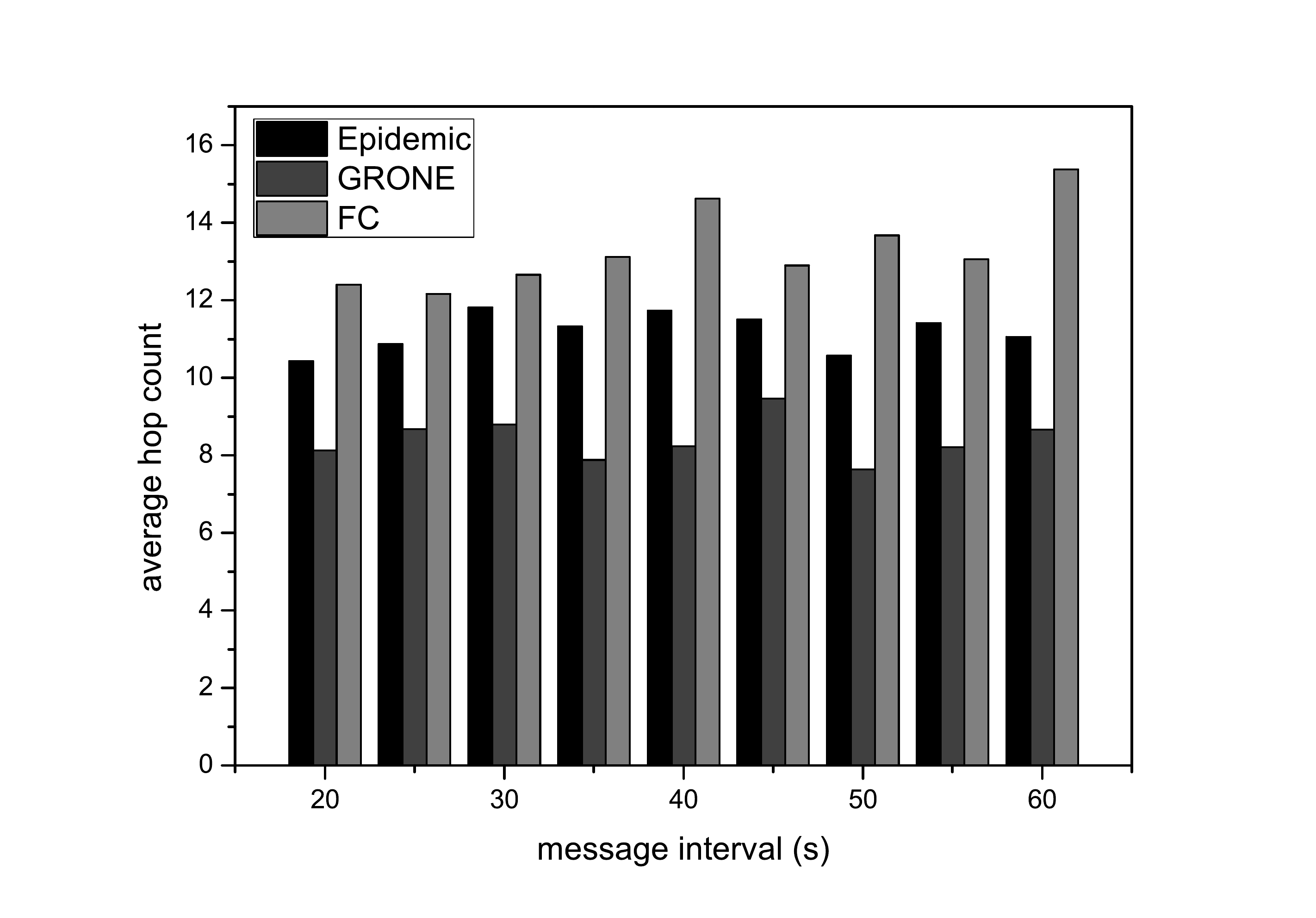}}
\subfigure[varying buffer\label{hop_vs_buffer}]
{\includegraphics[width=0.48\textwidth]{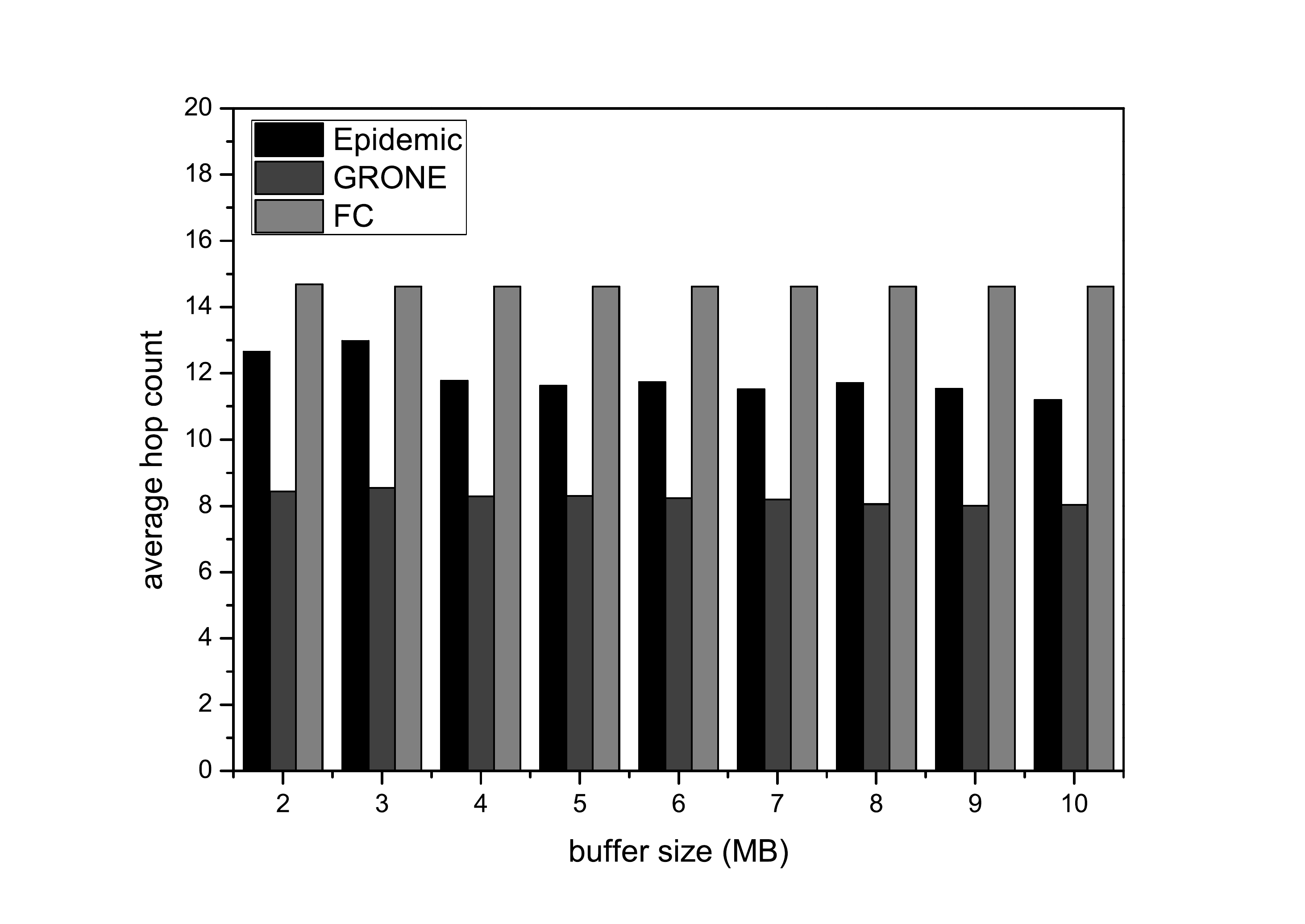}}\\
\subfigure[varying speed\label{hop_vs_speed}]
{\includegraphics[width=0.48\textwidth]{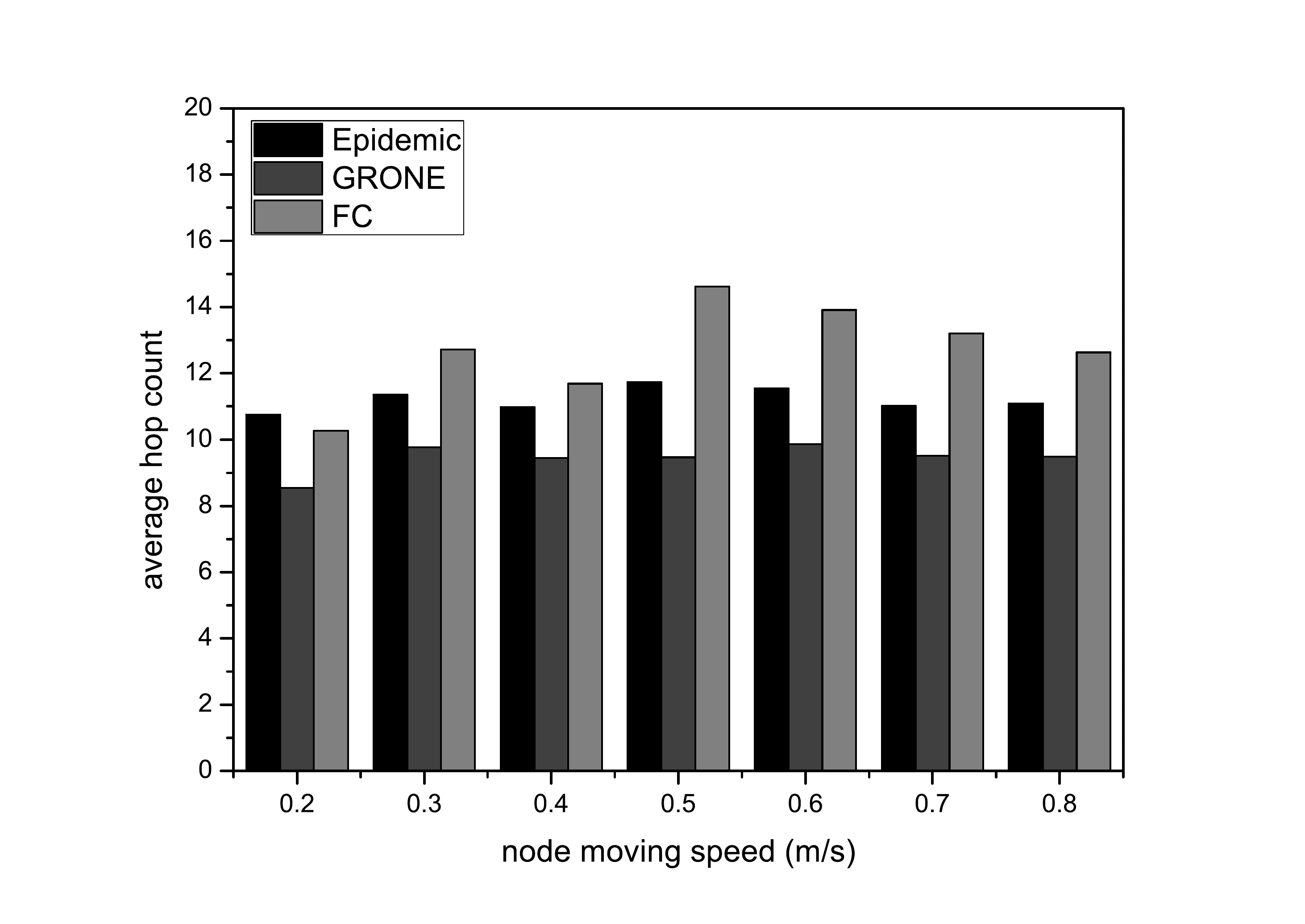}}
\subfigure[varying range\label{hop_vs_range}]
{\includegraphics[width=0.48\textwidth]{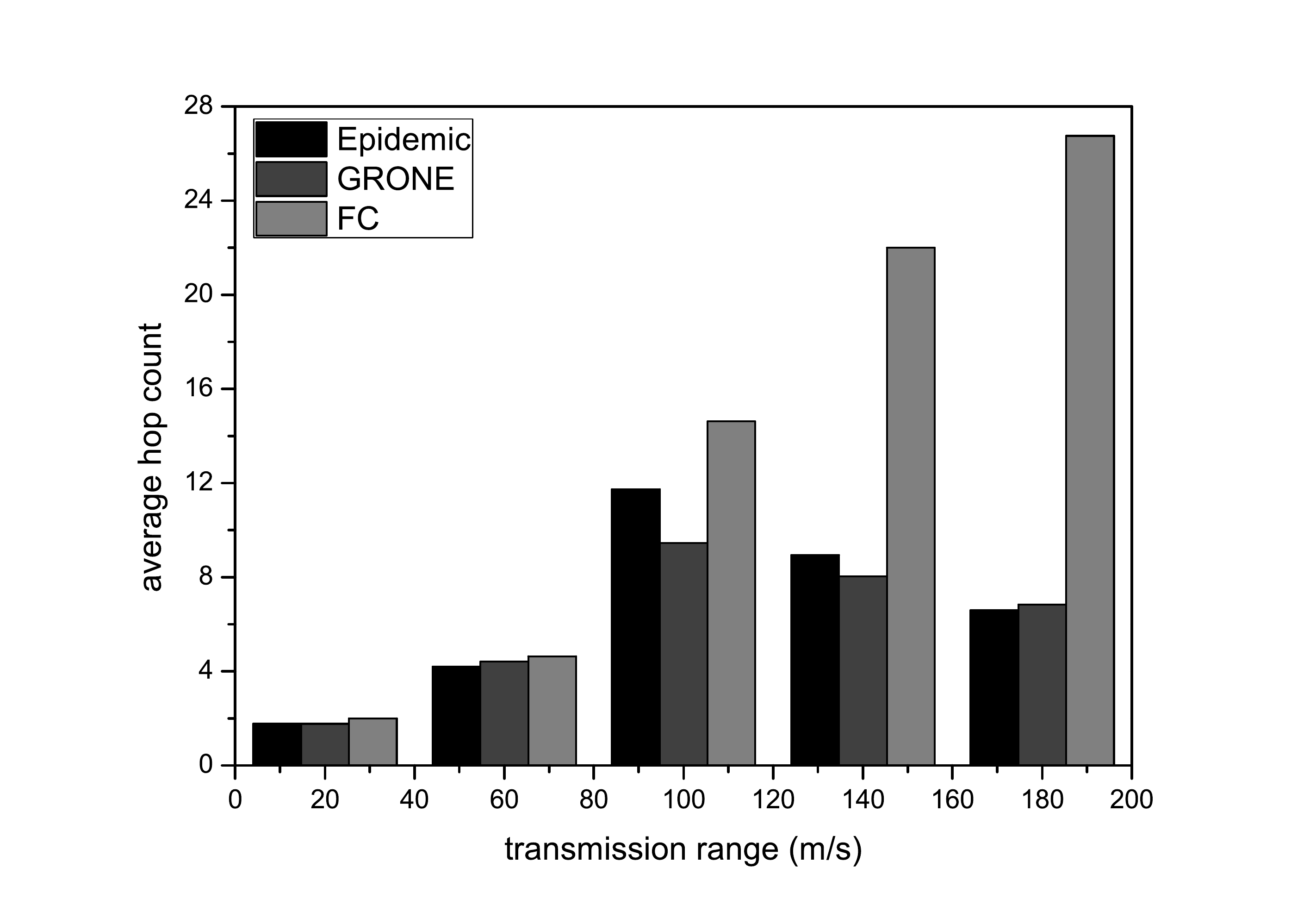}}
\caption{Average hop count vs message interval, buffer size,node moving speed
and transmission range.}
\label{hop}
\end{figure*}
However, things become very different when investigating Binary Spray \& Wait, for that the maximum possible hop count per message is limited by specific simulation configuration. Since we set the tickets to be 18 in our simulation for Binary Spray \& Wait protocol, the maximum hop for each message is less than $\lceil \log_{2}18\rceil=5$. As a result we only compare the other three algorithms in average hop count evaluation. 

Observed from the results shown in \figurename~\ref{hop_vs_interval}, \ref{hop_vs_buffer} and \ref{hop_vs_speed}, GRONE has the least average hop count 
among the three protocols, for the reason that it makes 
rational next hop choices based on node positions, thus 
making the message move purposefully toward the destination. 
Meanwhile, If the network 
resource is quite sufficient, Epidemic should perform at 
least as good as GRONE does. But in our simulation, the 
buffer resource is always set to be in a relatively constrained 
range like most actual network deployments, 
so that the incurred frequent message drop would introduce some 
extra hop count into the routing process. 
Furthermore, we can find that FC always have the highest average hop due to its intrinsic single copy and naive next-hop choice.

As illustrated in \figurename~\ref{hop_vs_interval}, we find 
that the message interval does not have great effect on the 
average hop count for all the three routing algorithms. The 
same situation shows up in \figurename~\ref{hop_vs_buffer}, 
and we find that FC maintains an  approximately constant average hop count in the case of sufficient buffer resource supply.

\figurename~\ref{hop_vs_speed} shows that the average hop 
count of FC is
more sensitive to node moving speed than the other two 
protocols. We can roughly conclude from the result that the 
average hop count for FC usually becomes larger along with the 
increase of node moving speed, because it always forwards 
the message blindly 
to the first met neighbor node that might not always 
be a suitable relay for the destination. 

In \figurename~
\ref{hop_vs_range}, when the transmission radius is comparatively short, the average hop count is low. 
However, in this case the delivery ratio is 
extremely low, as shown in \figurename~\ref{delivery_vs_range}, and thus we conclude that the overall performance of all protocols is terrible.
In this case, all the counted 
messages are usually delivered within small hop from source node, 
since that the network connectivity is too weak to get the message transferred further. For a message which needs at least a few hops to be forwarded before arriving at the destination, it has a high probability of being 
dropped before arriving at the destination due to the 
limited node buffer resource. Consequently, when the transmission
range is small, all these three algorithms have a relatively
low average hop count.
As the transmission range increases, the average hop 
count of both Epidemic and GRONE arise, since that 
there are more and more opportunities for each message to be transferred, thus increasing the number of counted messages.
However, when the transmission range becomes quite large, quite a number of messages could be delivered within only a
few hops from the source to destination, and thus the average hop count for Epidemic and GRONE decrease.
As for FirstContact protocol, it has a totally different behavior that the 
average hop count keeps increasing with network connectivity 
enhancing. The reason is the same as that mentioned before, i.e. the blind forwarding strategy might incur many unnecessary relay operations.

\figurename~\ref{overhead}, illustrates the overhead ratio performance for all the four routing protocols.
\begin{figure*}
\centering
\subfigure[varying interval\label{overhead_vs_interval}]
{\includegraphics[width=0.48\textwidth]{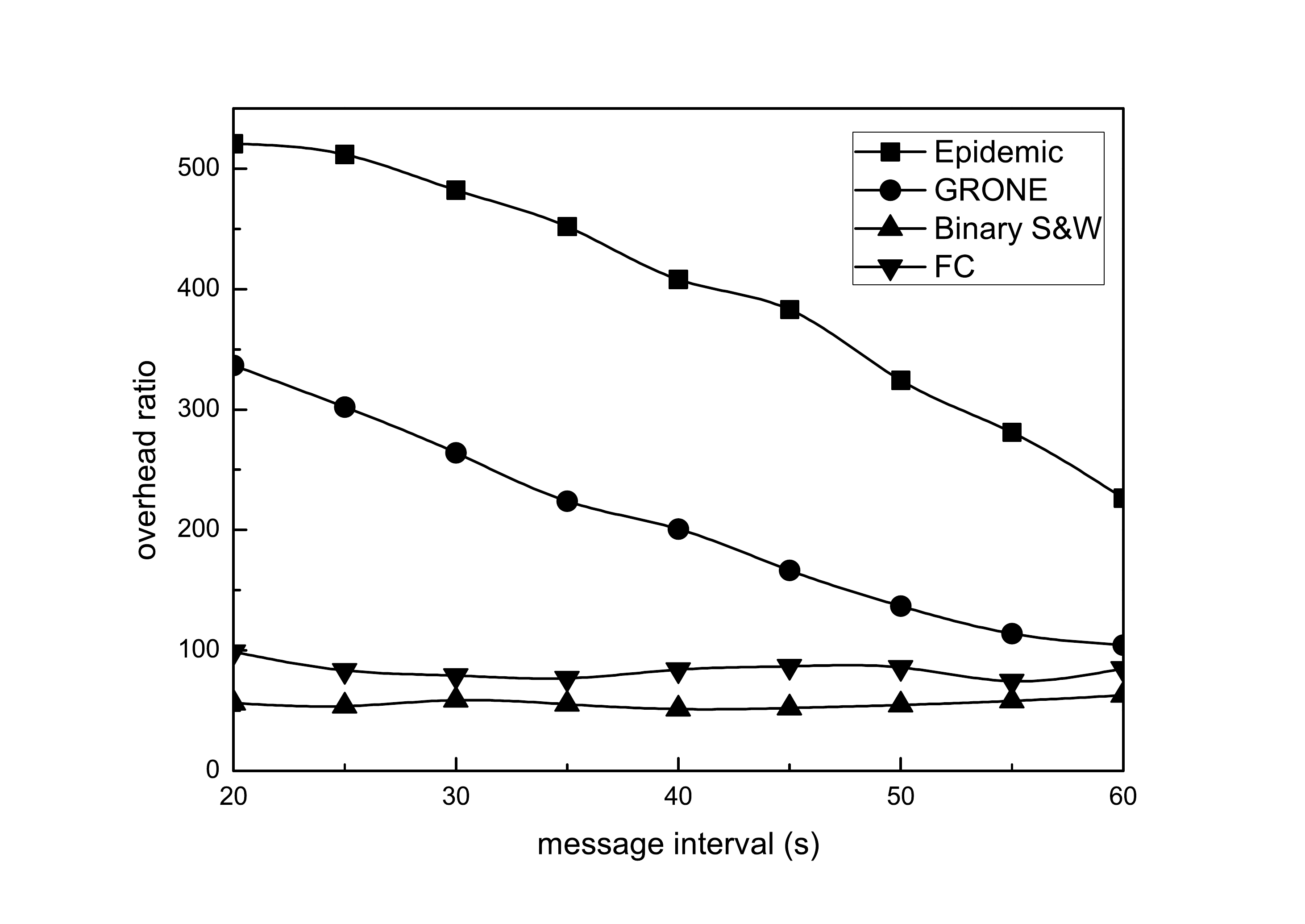}}
\subfigure[varying buffer\label{overhead_vs_buffer}]
{\includegraphics[width=0.48\textwidth]{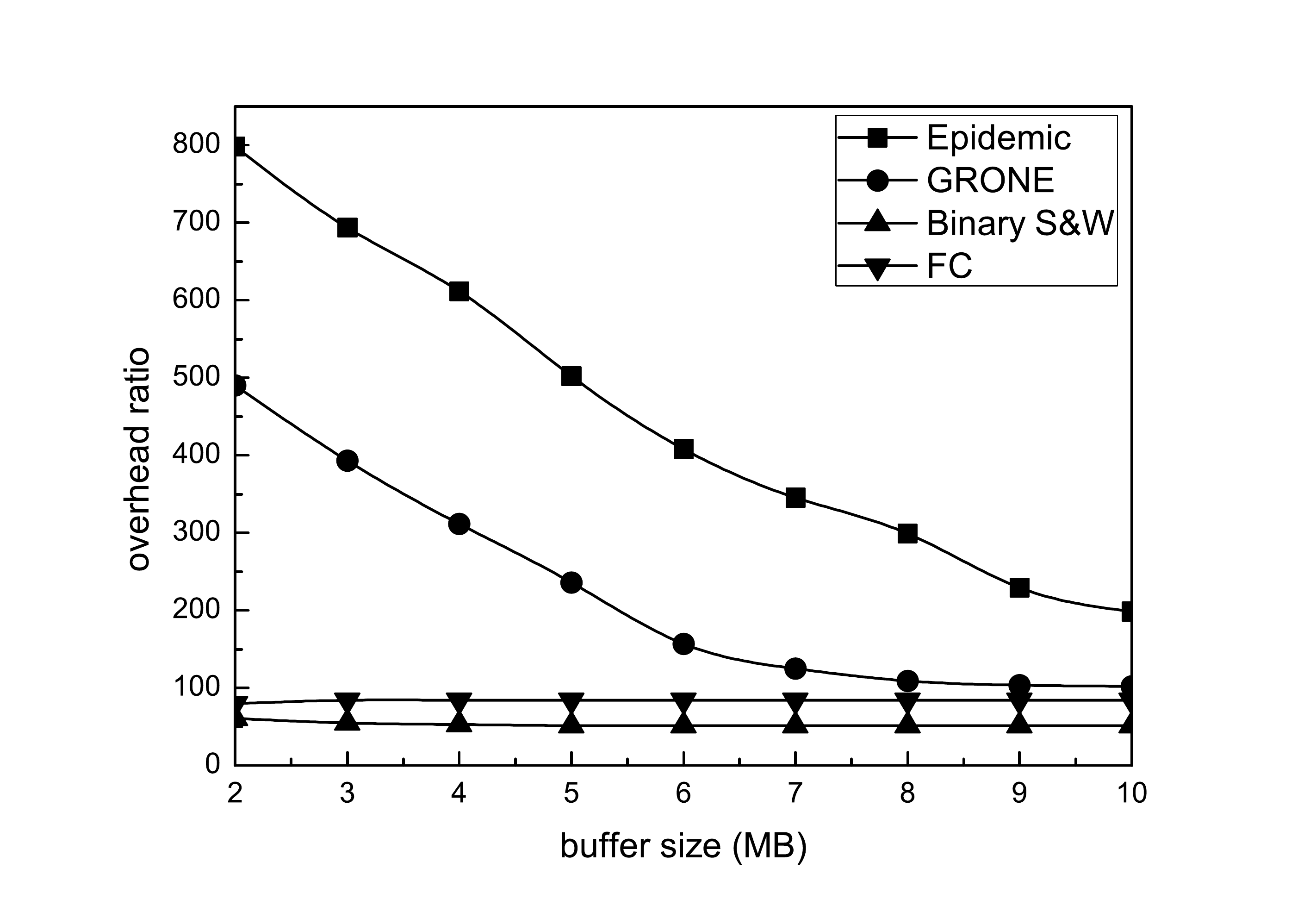}}\\
\subfigure[varying speed\label{overhead_vs_speed}]
{\includegraphics[width=0.48\textwidth]{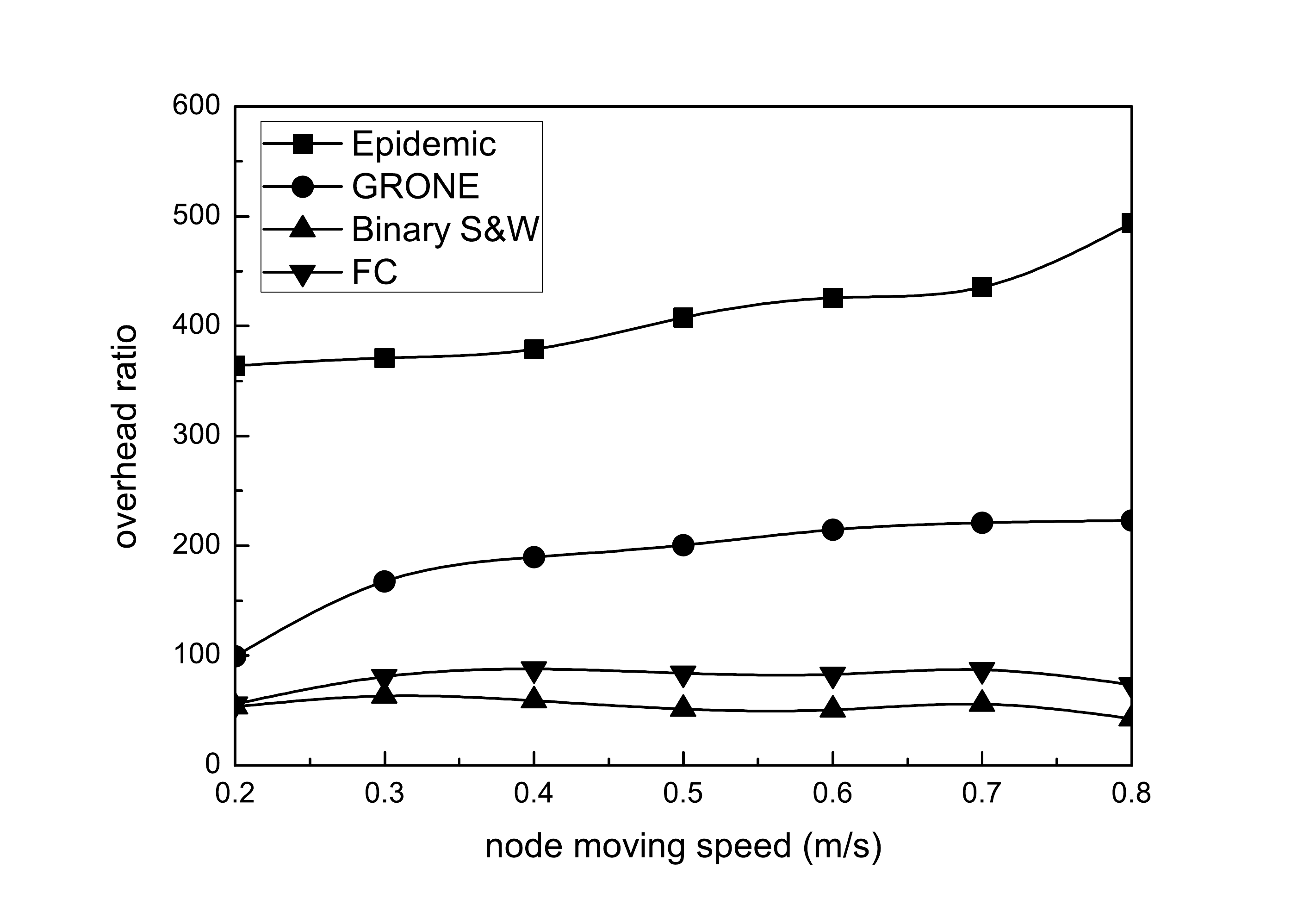}}
\subfigure[varying range\label{overhead_vs_range}]
{\includegraphics[width=0.48\textwidth]{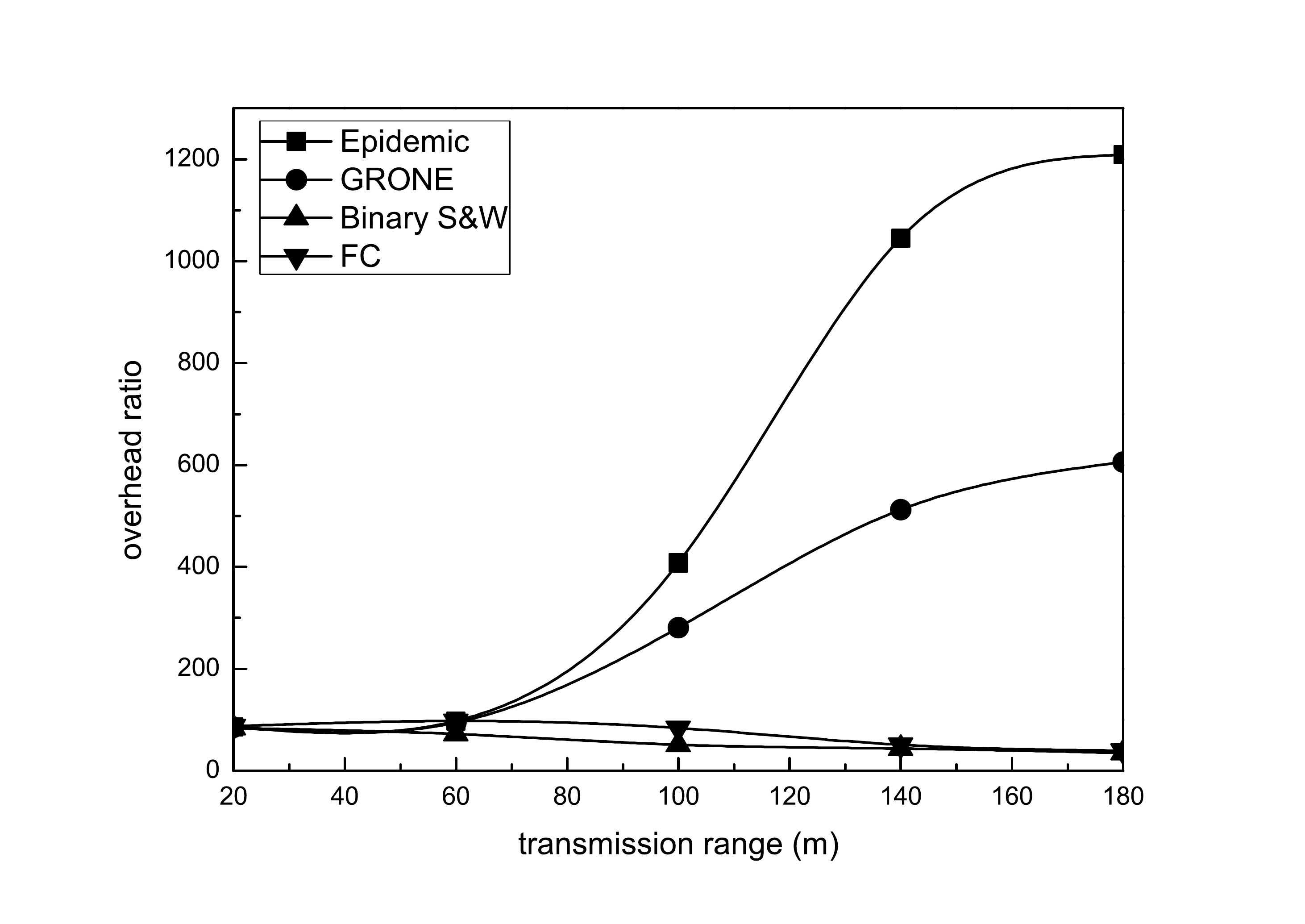}}
\caption{Overhead ratio vs message interval, buffer size, node moving speed
and transmission range.}
\label{overhead}
\end{figure*}
As defined in subsection 6.1, the overhead ratio is a criterion that reflects the efficiency of message transmission. 
We know that the Direct Delivery (DD) \cite{Grossglauser2002} has a zero overhead ratio, since $M^{relayed}$ is always strictly equal to $M^{delivered}$.  For Binary Spray \& Wait, the same forwarding strategy as DD is employed in the ``wait'' stage, and meanwhile, the ``spray'' stage limits the maximum number of copies for each message, and intuitively we can judge that Spray \& Wait should have a relatively low overhead ratio. FC is a single copy routing protocol and uses record vector to avoid routing loops, so $M^{relayed}\leq|maximum\ hop\ count|\cdot M^{created}$. Thus we can roughly estimate that the overhead ratio of FC from this inequality, and the result indicates that its overhead ratio is as low as that of Binary Spray \& Wait. 
Our analysis is verified by the simulation results shown in \figurename~\ref{hop_vs_interval}, \ref{hop_vs_buffer} and \ref{hop_vs_speed}, that Epidemic and GRONE have higher overhead ratio than Binary Spray \& Wait and FC. However, as shown in \figurename~\ref{delivery}, neither Binary Spray \& Wait nor FC has an acceptable delivery ratio. 

Let us focus on comparing our routing algorithm GRONE with Epidemic. 
From the previous analysis, we know that GRONE approximately has the same high delivery ratio as Epidemic. While from all the four subfigures in \figurename~\ref{overhead}, it is found that GRONE has a considerably lower overhead ratio than Epidemic, since it implements replication mechanism based on a utility function that in some sense reduces unnecessary forwarding operations. Besides, the message redundancy coping mechanism is also helpful to reduce the number of message copies and relay operations. In \figurename~\ref{overhead_vs_interval}, when message interval increases, the overhead ratio of both two protocols decrease, for the reason that the less generated messages lead to the less relay operations. \figurename~\ref{overhead_vs_buffer} shows that increasing buffer size is a good solution to cope with high overhead ratio, since the delivery ratio would increase, thus making the number of delivered messages arise. In \figurename~\ref{overhead_vs_speed}, we can conclude that when nodes have a higher moving speed, the overhead ratio would also arise. \figurename~\ref{overhead_vs_range} shows that when the connectivity of the whole network is strengthened, the overhead ratio of GRONE increases much lower than Epidemic. 

To evaluate the the stability of routing performance for GRONE, we vary the simulation time and the number of all nodes and then give the simulation results in \figurename~\ref{number_of_nodes}. 
\begin{figure}
\centering
\includegraphics[width=0.9\linewidth]{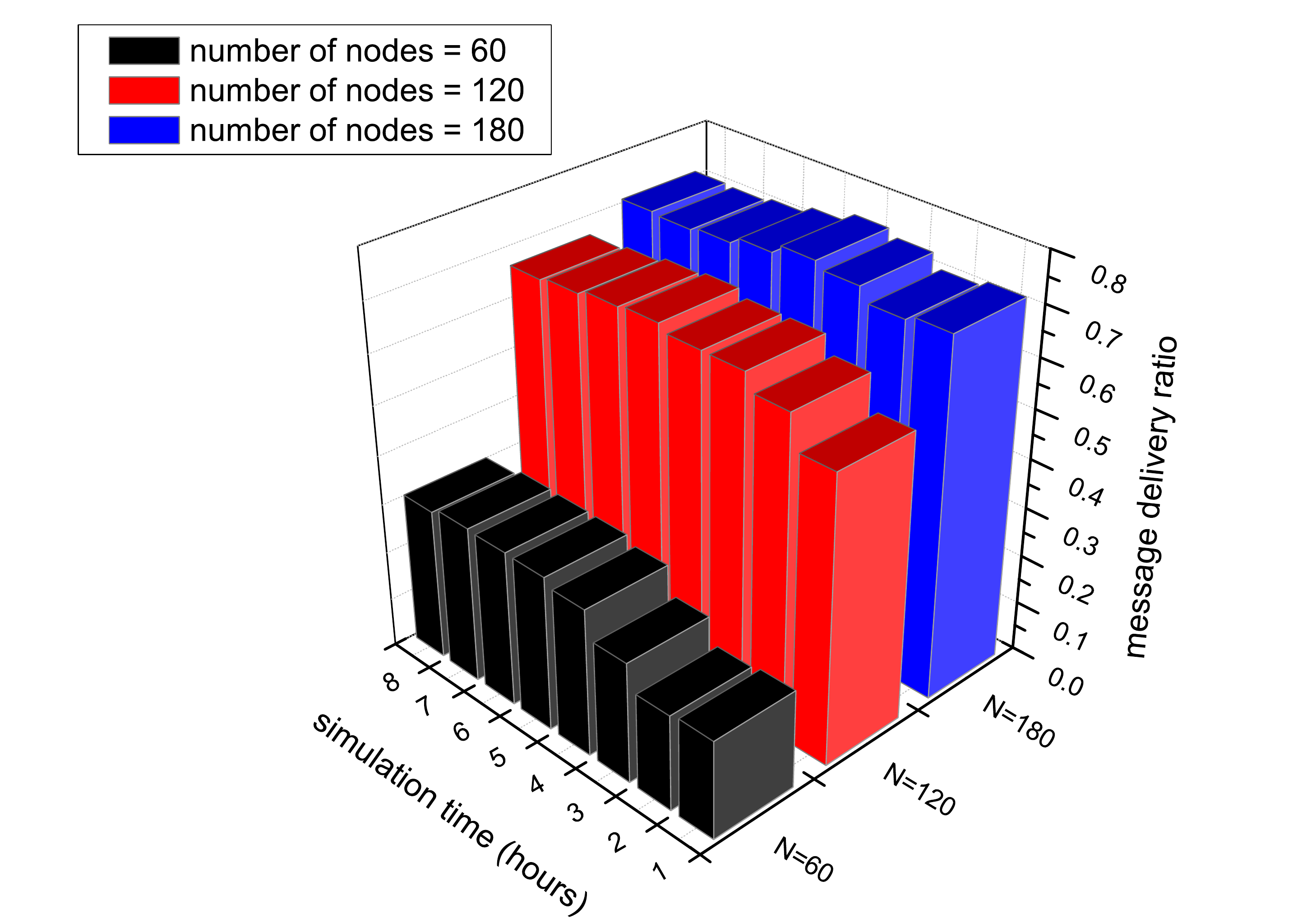}
\caption{Investigating message overhead ratio by varying number of 
nodes and simulation time}
\label{number_of_nodes}
\end{figure}
It seems that the duration of simulation time has trivial effect on the message delivery ratio of GRONE. Additionally, GRONE is more suitable to be implemented in the network where nodes are deployed not too sparsely. However, since the default node moving speed is set to be comparatively low, and if the connectivity of the network is too weak, it brings no surprise that routing protocols cannot achieve an ideal performance, and this has already been demonstrated in \figurename~\ref{delivery_vs_range}. Generally speaking, the stability of GRONE is moderate, and the performance is acceptable when node density is not too low.

Finally, we investigate the performance of GRONE by adjusting the node moving speed and node transmission range, as shown in \figurename~\ref{simulation_time}. 
\begin{figure}
\centering
\includegraphics[width=0.9\linewidth]{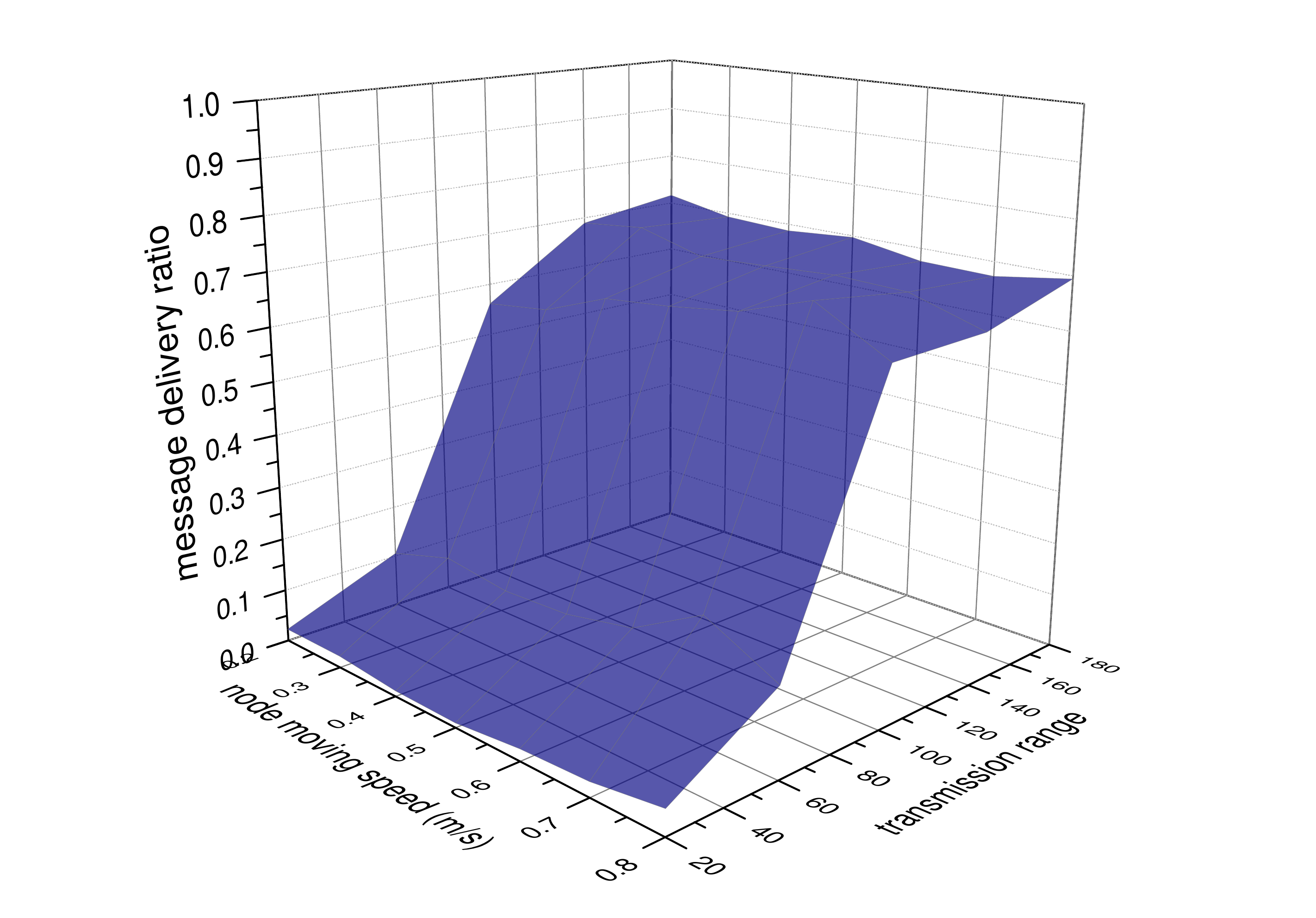}
\caption{Investigating message overhead ratio by varying node moving speed and transmission range}
\label{simulation_time}
\end{figure}
We find that the transmission range has a significant effect on the message delivery ratio of GRONE. The result is quite similar to that in \figurename~\ref{number_of_nodes}, which implies that GRONE is more applicable in a network with relatively high node density.

\section{Conclusion}
\label{con}
In this paper, we proposed a geographic routing protocol GRONE based on one-hop neighboring nodes information. By defining the utility function that relies on the position information of neighbors, GRONE achieves approximately as high delivery ratio as Epidemic while introducing significantly lower overheads to the network. Since we have no assumptions of the position of the destination node, a good way to maximize the delivery probability and lower the average hop count is to spread the message uniformly to each direction, and transmit them as far as possible. The utility function takes both of the two factors, transmission direction and relay distance, into consideration, and ensures that the messages roughly spread in a radiating manner. The simulations results indicate that GRONE has a lower average hop count than both Epidemic and FC. 

Besides, k-MDR, a criterion to evaluate the message redundancy, is defined, and by using it we prove that the 2-MDR for each message will keep increasing in GRONE. And for the purpose of coping with the message redundancy, we add the redundancy coping mechanism into GRONE. The simulation results show that the overheads ratio of GRONE is considerably lower than Epidemic. Besides, when the moving speed of nodes is relatively low, geographic routing protocol works quite well, as the position of each node keeps relatively stable.

Finally, GRONE has a simple implementation process since it only resorts one hop neighboring node information. In the future, we intend to explore a more precise evaluation criterion for relay nodes selection, which can better adapt to the network topology changes. Furthermore, we intend to add some message ferrying nodes to facilitate routing and schedule their mobility trajectory in order to achieve a higher delivery ratio.

\section{Acknowledgement}
This research was supported in part by Foundation research project of Qingdao Science and Technology Plan under Grant No.12-1-4-2-(14)-jch and Natural Science Foundation of Shandong Province under Grant No.ZR2013FQ022

\bibliographystyle{IEEEtran}


\end{document}